%% file: ms.tex
\documentclass[11pt]{article}

\usepackage[margin=4cm]{geometry}

\input{header.tex}

\newcommand{\DocumentTitle}{
  Unprovability results for clause set cycles
}

\author[1,2]{Stefan Hetzl}
\author[1,3]{Jannik Vierling}
\affil[1]{Vienna University of Technology\protect\\Institute of Discrete Mathematics and Geometry}
\affil[2]{\emailHetzl}
\affil[3]{\emailVierling}
\date{}
\title{\DocumentTitle}

\hypersetup{
 pdfauthor={Jannik Vierling},
 pdftitle={\DocumentTitle},
 pdfkeywords={},
 pdfsubject={},
 pdfcreator={}, 
 pdflang={English}}

\begin{document}

\maketitle

\begin{abstract}
  The notion of clause set cycle abstracts a family of methods for automated inductive theorem proving based on the detection of cyclic dependencies between clause sets.
  By discerning the underlying logical features of clause set cycles, we are able to characterize clause set cycles by a logical theory.
  We make use of this characterization to provide practically relevant unprovability results for clause set cycles that exploit different logical features.
\end{abstract}


\section{Introduction}



The subject of \ac{AITP} is a subfield of automated theorem proving, that aims at automating the process of finding proofs that involve mathematical induction.
The most prominent application of automated inductive theorem proving is the formal verification of hardware and software.
Another field of application of automated inductive theorem proving is the formalization of mathematical statements, where \ac{AITP} systems assist humans in formalizing statements by discharging lemmas automatically, suggest inductions \cite{nagashima2019}, or explore the theory \cite{johansson2014,valbuena2015}.


Finding a proof by mathematical induction essentially amounts to finding suitable induction formulas \cite{wong2018}.
This is a challenging task, because induction formulas have in general a higher syntactic complexity than the formula one wants to prove.
This phenomenon is commonly known as the \emph{non-analyticity} of induction formulas and can for example manifest itself in the number of free variables as well as the number of quantifier alternations of the induction formula.
Indeed, in the language of primitive recursive arithmetic there is a sequence of quantifier-free formulas whose proofs require induction formulas of unbounded quantifier complexity.
We refer to \cite{wong2018} for a precise exposition of the non-analyticity phenomenon.


A large variety of methods for automating mathematical induction has been developed.
Methods usually differ in the type of induction formulas they generate, the calculus they are integrated in, and other more technical parameters such as the degree of automation, the input encoding, semantics of datatypes, and so on.
For example there are methods based on term rewriting \cite{reddy1990}, theory exploration \cite{claessen2013a}, and integration into saturation-based provers \cite{kersani2013,kersani2014}, \cite{cruanes2015,cruanes2017}, \cite{echenim2020}, \cite{reger2019,hajdu2020}.

The current methodology in automated inductive theorem proving focuses on empirical evaluations of its methods.
A given method is usually evaluated on a set of benchmark problems such as \cite{claessen2015}.
Such an evaluation provides evidence about the strengths and weaknesses of a method but does not result in a systematic understanding of the underlying principles.
In particular, it is difficult to compare the methods with each other in terms of their logical strength and to provide explanations of the failures of a given method.


The work in this article is part of a research program that addresses this problem by formally analyzing \ac{AITP} systems in order to discern their underlying logical principles.
The analysis of an \ac{AITP} system typically begins by developing a suitable abstraction.
After that, the abstraction is simulated by a logical theory, whose properties can be investigated by applying powerful results and techniques from mathematical logic.
Analyzing families of methods in the uniform formalism of logic allows us not only to understand the strength of individual methods but also to compare the methods with each other.
Furthermore, approximating an \ac{AITP} system by a logical theory is a prerequisite to providing concrete and practically meaningful unprovability results.
These results are especially valuable because they allow us to determine the logical features that a given method lacks.
Thus, negative results drive the development of new and more powerful methods.


In \cite{hetzl2020} the authors of this article have introduced the notion of clause set cycle as an abstraction of the n-clause calculus \cite{kersani2013,kersani2014}---an extension of the superposition calculus by a cycle detection mechanism.
In particular, we have shown an upper bound on the strength of clause set cycles in terms of induction for \(\exists_{1}\) formulas and moreover that this bound is optimal with respect to the quantifier complexity of induction formulas.

In this article we continue this analysis of clause set cycles.
By discerning the logical features underlying the formalism of clause set cycles more precisely, we are able to provide an exact characterization of refutation by a clause set cycle in terms of a logical theory.
After that, we make use of the characterization of clause set cycles to provide practically meaningful clause sets that are not refutable by clause set cycles, but that are refutable by induction on quantifier-free formulas.
Hence, the results in this article settle in particular Conjecture~4.7 of \cite{hetzl2020}.
We provide unrefutability results that exploit different logical features of clause set cycles.
This allows to recognize features that are particularly restrictive.

In Section~\ref{sec:preliminaries} we will first introduce general notions and results about the logical setting that we use in this article.
In Section~\ref{sec:clause-set-cycles} we carry out the analysis of clause set cycles which culminates in Section~\ref{sec:unpr-clause-set} with two unprovability results for clause set cycles.
The two unprovability results exploit different logical features of clause set cycles.
The proof of the first unprovability result is straightforward, whereas the second unprovability result relies on a more involved independence result in the setting of linear arithmetic whose proof is carried out in Section~\ref{sec:idemp-line-arithm}.
\section{Preliminary definitions}
\label{sec:preliminaries}
In this section we introduce some definitions that we will use throughout the article.
In Section \ref{sec:formulas-theories} we will briefly describe the logical formalism.
In Section \ref{sec:linear-arithmetic} we describe the setting of formal linear arithmetic, in which we will formulate in Section~\ref{sec:unpr-clause-set} a family of clause sets that are refutable by open induction but that are not refutable by clause set cycles.
\subsection{Formulas, theories, and clauses}
\label{sec:formulas-theories}
We work in a setting of classical logic with equality, that is, the logic provides besides the usual logical symbols a binary infix predicate symbol \({=}\) representing equality.
A first-order language \(L\) is a set of predicate symbols and function symbols together with their respective arities.
Let \(S\) be a predicate or function symbol, then we write \(S/n\) to indicate that \(S\) has arity \(n\).
Terms, atoms, and formulas are constructed as usual from function symbols, variable symbols, the logical connectives \(\neg\), \(\wedge\), \(\vee\), \(\rightarrow\), \(\leftrightarrow\), and the quantifiers \(\exists\) and \(\forall\).
A ground term is a term that does not contain variables.
The set of all \(L\) formulas is denoted by \(\mathcal{F}(L)\).
A sentence is a formula that does not contain free variables.
Let \(x_{1}, \dots, x_{n}\) be variables, \(t_{1}, \dots, t_{n}\) a terms, and \(\varphi\) a formula, then \(\varphi[x_{1}/t_{1}, \dots, x_{n}/t_{n}]\) denotes the simultaneous substitution of \(x_{i}\) by \(t_{i}\) for \(i = 1, \dots, n\) in \(\varphi\).

In this article we are more interested in the axioms of a theory, rather than the deductive closure of these axioms.
Hence, we define a theory as a set axioms and manipulate the deductive closure by means of the first-order provability relation.
\begin{definition}[Theories and provability]
  \label{def:15}
  A theory \(T\) is a set of sentences called the axioms of \(T\).
  Let \(T, U\) be theories, then by \(T + U\) we denote the theory axiomatized by \(T \cup U\).
  Let \(\varphi\) be a formula, then we write \(T \vdash \varphi\) if \(\varphi\) is provable in first-order logic from the axioms \(T\).
  Let \(\Gamma\) be a set of formulas, then we write \(T \vdash \Gamma\) if \(T \vdash \gamma\) for each \(\gamma \in \Gamma\).
  Furthermore, we write \(T \equiv U\), if \(T \vdash U\) and \(U \vdash T\).
\end{definition}
Let \(T\) be a theory and \(\varphi\) a formula, then we write \(T + \varphi\) to denote the theory axiomatized by the axioms of \(T\) and the universal closure of \(\varphi\).
In this article we will be particularly interested in formulas with a restricted number of quantifier alternations.
\begin{definition}
  \label{def:18}
  We say that a formula \(\varphi\) (possibly containing free variables) is \(\forall_{0}\) or \(\exists_{0}\) if \(\varphi\) is quantifier-free.
  Moreover, we say that a formula \(\varphi(\vec{z})\) is \(\forall_{k+1}\) (\(\exists_{k+1}\)) if it is of the form \(\Forall{\vec{x}}{\psi(\vec{x},\vec{z})}\) (\(\Exists{\vec{x}}{\psi(\vec{x},\vec{z})}\)) and \(\psi\) is \(\exists_{k}\) (\(\forall_{k}\)).
  Let \(L\) be a first-order language, then by \(\Open(L)\), \(\exists_{k}(L)\), and \(\forall_{k}(L)\), we denote the quantifier-free formulas, \(\exists_{k}\) formulas, and the \(\forall_{k}\) formulas of the language \(L\).
A theory is said to be \(\exists_{k}\) (\(\forall_{k}\)) if all of its axioms are \(\exists_{k}\) (\(\forall_{k}\)) sentences.
\end{definition}
Clause sets are an alternative representation of \(\forall_{1}\) formulas, that is preferred by automated theorem provers because of its uniformity.
\begin{definition}[Literals, clauses, clause sets]
  \label{def:11}
  Let \(L\) be a first-order language.
  By an \(L\) literal we understand an \(L\) atom or the negation of an \(L\) atom.
  An \(L\) clause is a finite set of \(L\) literals.
  An \(L\) clause set is a set of \(L\) clauses.
  Whenever the language \(L\) is clear from the context, we simply speak of atoms, literals, clauses and clause sets.
\end{definition}
We will now recall some basic model-theoretic concepts.
\begin{definition}
  \label{def:19}
  Let \(L\) be a first-order language, then \(L\) structures and the first-order satisfaction relation \(\models\) are defined as usual.
  Let \(L' \subseteq L\) be a first-order language and \(M\) an \(L\) structure, then by \(M|_{L'}\) we denote the \(L'\) reduct of \(M\).
Let \(M\) be an \(L\) structure, then we write \(b \in M\) to express that \(b\) is an element of the domain of \(M\).
Formulas and clauses are interpreted as usual.
In particular, a clause is interpreted as the universal closure of the disjunction of its literals.
Let \(\Delta\) be a set of \(L\) formulas and \(L\) clauses, then \(M \models \Delta\) if \(M \models \delta\) for each \(\delta \in \Delta\).
\end{definition}
Let us conclude this section by introducing some notation to manipulate clauses and clause sets.
\begin{definition}
  \label{def:20}
  By \(\cnf\) we denote a fixed function that assigns to every \(\forall_{1}\) sentence \(\varphi\), a clause set \(\cnf(\varphi)\) over the language of \(\varphi\) such that \(\varphi\) and \(\cnf(\varphi)\) are logically equivalent.
  Let \(\Gamma\) be a set of \(\forall_{1}\) sentences, then we define \(\cnf(\Gamma) \coloneqq \bigcup_{\gamma \in \Gamma}\cnf(\gamma)\).
  Furthermore, by \(\cnf^{-1}\) we denote a fixed function that assigns to every clause set \(\mathcal{C}\) a \(\forall_{1}\) sentence \(\cnf^{-1}(\mathcal{C})\) over the language of \(\mathcal{C}\) such that \(\mathcal{C}\) and \(\cnf^{-1}(\mathcal{C})\) are logically equivalent.
\end{definition}
\begin{lemma}
  \label{lem:disjunction_of_clause_sets}
  Let \(\mathfrak{C}\) be a finite set of clause sets, then there exists a clause set \(\mathcal{C}'\) such that \(M \models \mathcal{C}'\) if and only if there exists \(\mathcal{C} \in \mathfrak{C}\) such that \(M \models \mathcal{C}\).
\end{lemma}
\begin{proof}
  Let \(\mathfrak{C} = \{\mathcal{C}_{1}, \dots, \mathcal{C}_{n}\}\), now let \(\psi \coloneqq \bigvee_{i = 1}^{n}\cnf^{-1}(\mathcal{C}_{i})\).
  Then \(\psi\) is logically equivalent to a \(\forall_{1}\) sentence \(\psi'\).
  Now we define \(\mathcal{C}' \coloneqq \cnf(\psi')\).
  It is clear that \(M \models \mathcal{C}'\) if and only if there exists \(i \in \{ 1, \dots, n\}\) such that \(M \models \mathcal{C}_{i}\).
\end{proof}
\subsection{Induction and formal arithmetic}
\label{sec:linear-arithmetic}
In this section we introduce some basic notions about induction and formal arithmetic.
In particular we introduce the setting of linear arithmetic in which we formulate an unrefutability result for clause set cycles in Section~\ref{sec:unpr-clause-set}.

Inductive theorem provers customarily work in a many-sorted setting with a notion of inductive datatypes encompassing at least the natural numbers, lists, trees, and sometimes even more complicated types such as mutually recursive datatypes.
However, working in such a general setting is notationally tedious.
Moreover, all the phenomena we are interested in can already be observed over the natural numbers.
Hence, we restrict ourselves in this article to a one-sorted setting over the natural numbers.
By \(0/0\) and \(s/1\) we denote function symbols that represent the number zero and the successor function on natural numbers, respectively.
We fix some abbreviations.
Let \(n\) be a natural number and \(t\) a term, then \(s^{n}(t)\) denotes the term \(\underbrace{s(\cdots s}_{\text{\(n\) times}}(t) \cdots)\) and \(\numeral{n}\) denotes the term \(s^{n}(0)\).
Furthermore, let \(+\) be a binary infix function symbol representing addition of natural numbers, then the notation \(n \cdot t\) for the multiplication of the term \(t\) by the constant \(n\) is defined inductively by \(0 \cdot t = 0\) and \((i + 1) \cdot t = t + (i \cdot t)\).
\begin{definition}
  \label{def:10}
  Let \(\varphi(x, \vec{z})\) be a formula, then \(I_{x}\varphi\) denotes the formula
  \[
    \varphi(0, \vec{z}) \wedge \Forall{x}{(\varphi(x,\vec{z}) \rightarrow \varphi(s(x), \vec{z}))} \rightarrow \Forall{x}{\varphi(x, \vec{z})}.
  \]
  In the definition above, we call \(\varphi\) the induction formula, \(x\) the induction variable, and \(\vec{z}\) the induction parameters.
  Let \(\Gamma\) be a set of formulas, then the theory \(\IND{\Gamma}\) is axiomatized by the universal closure of the formulas \(I_{x}\gamma\) with \(\gamma \in \Gamma\).
\end{definition}
If induction is carried out on formulas without induction parameters, we speak of parameter-free induction.
A notion related to the induction scheme is that of inductivity in a theory.
\begin{definition}
  \label{def:22}
  Let \(T\) be a theory.
  A formula \(\varphi(x,\vec{z})\) is \(T\)-inductive in \(x\) if \(T \vdash \varphi(0, \vec{z})\) and \(T \vdash \varphi(x, \vec{z}) \rightarrow \varphi(s(x),\vec{z})\).
  Whenever the induction variable \(x\) is clear from the context we simply say that \(\varphi\) is inductive in \(T\).
\end{definition}
Let us now introduce the setting of linear arithmetic.
This setting has the advantage of being sufficiently complex to provide interesting independence results while still having straightforward model theoretic properties.
\begin{definition}[Language of linear arithmetic]
  \label{def:17}
  The function symbol \(p/1\) represents the predecessor function on natural numbers and the infix function symbol \(+/2\) represents the addition of natural numbers.
  The language \(\LLA\) of linear arithmetic is \(\{ 0,s,p,+ \}\).
\end{definition}
By \(\mathbb{N}\) we denote the set of natural numbers as well as the \(\LLA\) structure whose domain is the set of natural numbers and that interprets the symbols \(0\), \(s\), \(+\) naturally and interprets the symbol \(p\) by \(p^{\mathbb{N}}(0) = 0\) and \(p^{\mathbb{N}}(n + 1) = n\) for all \(n \in \mathbb{N}\).
Analogously, we denote by \(\mathbb{Z}\) the set of integers and the \(\LLA\) structure whose domain consists of the integers and that interprets all symbols naturally.
In particular, \(\mathbb{Z}\) interprets the symbol \(p\) as the function \(x \mapsto x - 1\).
All the theories of linear arithmetic that we will work with are extensions of the following base theory.
\begin{definition}
  \label{def:minimal_arithmetic}
  The \(\LLA\) theory \(\BLA\) is axiomatized by the universal closure of the formulas
  \begin{align}
    s(0) & \neq 0, \label{Q_1} \tag{A1} \\
    p(0) & = 0, \label{ax:p:0} \tag{A2} \\
    p(s(x)) & = x, \label{ax:p:s} \tag{A3} \\
    x + 0 & = x, \label{Q_3}\tag{A4} \\
    x + s(y) & = s(x + y). \label{Q_4}\tag{A5}
  \end{align}
\end{definition}
In the following we will recall some basic properties of the theory \(\BLA\) and its extension by induction for quantifier-free formulas.
Clearly, we have \(\mathbb{N} \models \BLA\) and \(\mathbb{Z} \models \axiomA{1} + \axiomA{3} + \axiomA{4} + \axiomA{5}\), but \(\mathbb{Z} \not \models \axiomA{2}\), since \(\mathbb{Z} \models p(0) = -1\).
\begin{lemma}
  \label{lem:3}
  \begin{enumerate}[label=(\roman*),ref =\ref{lem:3}{.}{\it (\roman*)}]
  \item \(\axiomA{1} + \axiomA{2} + \axiomA{3} \vdash s(x) \neq 0\). \label{lem:3:1}
  \item \(\axiomA{3} \vdash s(x) = s(y) \rightarrow x = y\). \label{lem:3:2}
  \end{enumerate}
\end{lemma}
\begin{proof}
  For {\it (i)} assume that there exists \(x\) such that \(s(x) = 0\), then by \ref{ax:p:0} and \ref{ax:p:s} we have \(x = p(s(x)) = p(0) = 0\). Thus, \(s(0) = 0\) which contradicts \ref{Q_1}.
  For {\it (ii)} assume \(s(x) = s(y)\), then we have \(p(s(x)) = p(s(y))\) and by \eqref{ax:p:s} we obtain \(x = y\).
\end{proof}
\begin{definition}
  \label{def:21}
  Let \(T\) be an \(\LLA\) theory.
  We say \(T\) is sound if \(\mathbb{N} \models T\).
  Furthermore, \(T\) is \(\exists_{1}\)-complete if \(\mathbb{N} \models \sigma\) implies \(T \vdash \sigma\) for all \(\exists_{1}\) \(\LLA\) sentences \(\sigma\).
\end{definition}
\begin{lemma}
  \label{lem:33}
  Let \(t\) be an \(\LLA\) ground term, then there exists \(k \in \mathbb{N}\) such that \(\BLA \vdash t = \numeral{k}\).
\end{lemma}
\begin{proof}
  Proceed by induction on the structure of the term.
\end{proof}
\begin{lemma}
  \label{lem:34}
  Let \(\varphi\) be a quantifier-free \(\LLA\) sentence, then either \(\BLA \vdash \varphi\) or \(\BLA \vdash \neg \varphi\).
\end{lemma}
\begin{proof}
  By a straightforward induction on the structure of the sentence \(\varphi\).
  The only interesting case is the case where \(\varphi\) is an atom \(t_{1} = t_{2}\).
  By Lemma~\ref{lem:33} there exist \(k_{1}, k_{2} \in \mathbb{N}\) such that \(\BLA \vdash t_{1} = t_{2} \leftrightarrow \numeral{k_{1}} = \numeral{k_{2}}\).
  If \(k_{1} = k_{2}\) we apply reflexivity.
  Otherwise, we apply Lemma~\ref{lem:3:2} repeatedly and finally we use \ref{lem:3:1}.
\end{proof}
\begin{lemma}
  \label{lem:14}
  The theory \(\BLA\) is \(\exists_{1}\)-complete.
\end{lemma}
\begin{proof}
  By Lemma~\ref{lem:34}, \(\BLA\) is also complete for quantifier-free sentences.
  Assume that \(\mathbb{N} \models \Exists{\vec{x}}{\varphi(\vec{x})}\), where \(\varphi\) is quantifier-free.
  Then there are \(n_{1},\dots, n_{k}\) such that \(\mathbb{N} \models \varphi(\numeral{n}_{1}, \dots, \numeral{n_{k}})\).
  Therefore, \(\BLA \vdash \varphi(\numeral{n_{1}}, \dots, \numeral{n_{k}})\), thus, \(\BLA \vdash \Exists{\vec{x}}{\varphi(\vec{x})}\).
\end{proof}
\begin{lemma}
  \label{lem:29}
  The theory \(\BLA + \IND{\Open(\LLA)}\) proves the following formulas
  \begin{gather*}
    x = 0 \vee x = s(p(x)), \tag{\(\axiomB{1}\)} \label{B1} \\
    x + y = y + x, \tag{\(\axiomB{2}\)} \label{B2} \\
    x + (y + z) = (x + y) + z, \tag{\(\axiomB{3}\)} \label{B3} \\
    x + y = x + z \rightarrow y = z. \tag{\(\axiomB{4}\)} \label{B4}
  \end{gather*}
\end{lemma}
\begin{proof}
  Routine.
\end{proof}
\begin{definition}
  \label{def:9}
  The theory \(\BLAPrime\) is axiomatized by \eqref{Q_1}--\eqref{Q_4} and \eqref{B1}--\eqref{B4}.
\end{definition}
\begin{theorem}[{\cite{shoenfield1958}}]
  \label{thm:shoenfields_theorem}
  \(\BLA + \IND{\Open(\LLA)} \equiv \BLAPrime\).
\end{theorem}
\section{Analysis of clause set cycles}
\label{sec:clause-set-cycles}
In this section we carry out an analysis of the formalism of refutation by a clause set cycle.
In Section~\ref{sec:definition-bound} we define clause set cycles and recall some basic properties as well as some results from \cite{hetzl2020}.
After that, we will provide in Section~\ref{sec:refin-bound-theory} a characterization of clause set cycles in terms of a logical theory with induction.
Finally, in Section~\ref{sec:unpr-clause-set} we will use this characterization and an independence result, that will be proved in Section~\ref{sec:idemp-line-arithm}, to obtain concrete and practically meaningful unrefutability results for clause set cycles.

\subsection{Clause set cycles}
\label{sec:definition-bound}
Refutation by a clause set cycle is a formalism introduced in \cite{hetzl2020} by the authors of this article to describe abstractly the inductive arguments that take place in the n-clause calculus \cite{kersani2013,kersani2014}.
The n-clause calculus is an extension of the superposition calculus by a mechanism that detects cyclic dependencies between the derived clauses.
These cyclic dependencies correspond to arguments by infinite descent and thus establish the inductive unsatisfiability of a set of clauses.
The notion of refutation by a clause set cycle abstracts the underlying superposition calculus and the detection of the cycle in that proof system and therefore extracts the essence of the arguments by infinite descent that may appear in refutations by the n-clause calculus.

Since all the variables occurring in clauses are implicitly universally quantified, a clause set does not have a free variable on which we can carry out an argument by induction.
Instead we will rely on a special constant symbol \(\eta\), on which arguments by infinite descent will take place.
This is in analogy to the special constant \(\mathrm{n}\) that is used by the n-clause calculus for the same purpose, see \cite{kersani2013}.
The constant \(\eta\) can be thought of as a Skolem constant, that is selected before a refutation is attempted.
In particular, clauses may of course contain other Skolem symbols besides \(\eta\).

Carrying out arguments by infinite descent (or induction) only on positions of constants is unsurprisingly very restricting (see Corollary~\ref{cor:3}).
Since clause set cycles are used as an abstraction of the inductive cycles of the n-clause calculus, we did not extend the formalism to allow arguments to take place in more varied positions.
The logical characterization that we give in Section~\ref{sec:refin-bound-theory} makes considering such extensions easier.
In particular, the main unprovability result of this article, Corollary~\ref{cor:6}, does not rely on this restriction.
A method that lifts this restriction has been proposed in \cite{echenim2020}.

\begin{remark}
  \label{rem:1}
  In the literature \cite{kersani2013,kersani2014,hetzl2020} a constant such as \(\eta\) is usually called a parameter. In order to avoid confusion with induction parameters in the sense of Definition~\ref{def:10} we will not use this designation.
\end{remark}
Let \(\mathcal{C}\) be a clause set possibly containing \(\eta\), then we write \(\mathcal{C}(\eta)\) to indicate all the occurrences of \(\eta\) in \(\mathcal{C}\).
Let furthermore \(t\) be a term, then \(\mathcal{C}(t)\) denotes the clause set obtained by replacing all the occurrences of \(\eta\) in \(\mathcal{C}\) by \(t\).
\begin{definition}[Refutation by a clause set cycle]
  \label{def:1}
  Let \(L\) be a first-order language. A finite \(L \cup \{ \eta \}\) clause set \(\mathcal{C}(\eta)\) is called an \(L\) clause set cycle if it satisfies the following conditions
  \begin{align}
    \mathcal{C}(s(\eta)) & \models \mathcal{C}(\eta), \label{C1} \tag{C1} \\
    \mathcal{C}(0)  & \models \bot. \label{C2} \tag{C2}
  \end{align}
  Let \(\mathcal{D}(\eta)\) be an \(L \cup \{ \eta \}\) clause set, then \(\mathcal{D}(\eta)\) is refuted by an \(L\) clause set cycle \(\mathcal{C}(\eta)\) if
  \begin{equation}
    \mathcal{D}(\eta) \models \mathcal{C}(\eta). \label{RC} \tag{C3}
  \end{equation}
\end{definition}
A clause set cycle represents an argument by infinite descent in the following sense.
Suppose there is an \(L \cup \{ \eta\}\) structure \(M\) with \(D(M) = \mathbb{N}\) such that \(M \models \mathcal{C}(\eta)\).
By \eqref{C2} we have \(\eta^{M} > 0\).
Now let \(m \in \mathbb{N}\), then we denote by \(M[\eta \mapsto m]\) the \(L \cup \{ \eta \}\) structure with the same domain as \(M\), that interprets all non-logical symbols except \(\eta\) as \(M\), and interprets \(\eta\) as \(m\).
Then we have \(M[\eta \mapsto \eta^{M} - 1] \models \mathcal{C}(s(\eta))\) and by \eqref{C1} we now obtain \(M[\eta \mapsto \eta^{M} - 1] \models \mathcal{C}(\eta)\).
Hence, we obtain an finite strictly descending sequence of natural numbers \(m\) such that \(M[\eta \mapsto m] \models \mathcal{C}(\eta)\).
This is impossible, hence \(M \not \models \mathcal{C}(\eta)\).

\begin{remark}
  \label{rem:2}
  In the literature cycles on clause sets are usually equipped with parameters that control the offset and the descent step size and thus permit a more flexible usage of the cycles (see for example \cite{kersani2013}).
  In Definition~\ref{def:2} we shall consider clause set cycles with parameters inspired by the parameters found in the cycles of the n-clause calculus.
  After that, we show in Proposition~\ref{pro:refutation_by_csc_with_parameters_is_not_more_powerful} that such parameters do not make the system more powerful.
  In particular, \cite{hetzl2020} uses a slightly different notation.
  A refutation by a clause set cycle in \cite{hetzl2020} corresponds to a refutation by a \((1,0)\)-clause set cycle with external offset \(i \in \mathbb{N}\) in the sense of Definition~\ref{def:2}.
  Hence, by Proposition~\ref{pro:refutation_by_csc_with_parameters_is_not_more_powerful} the notion of refutation by clause set cycle used in \cite{hetzl2020} is exactly as powerful as the more elegant notion of refutation by a clause set cycle used in this article.
\end{remark}

Clause set cycles could be integrated into a saturation-based prover by carrying out the saturation process as usual and by detecting a clause set cycles among the clauses derived so far, thus satisfying Condition~\ref{RC} with respect to the set of generated clauses.
The detection of a clause set cycle could for example make use of the derivation relation generated by the prover in order to detect the Conditions~\ref{C1} and \ref{C2}.
The detection of a clause set cycle, then provides the inductive unsatisfiability of the clauses generated and therefore ends the refutation.
This is essentially how the n-clause calculus described in \cite{kersani2013,kersani2014} operates.

Let us now consider an example of a refutation by a clause set cycle.
\begin{definition}
  \label{def:26}
  By \(\mathsf{C}(\eta)\) we denote the \(\LLA \cup \{\eta\}\) clause set
  \begin{align*}
    \cnf(\BLA + \axiomB{2}) \cup \{ \{ \eta \neq x + x \}, \{ \eta \neq s(x + x)\} \}.
  \end{align*}
\end{definition}
\begin{example}
  \label{ex:1}
  Intuitively, the clause set \(\mathsf{C}(\eta)\) asserts the existence of an element \(\eta\), which is neither even nor odd.
  We will now show that \(\mathsf{C}(\eta)\) is a clause set cycle.
  
  We start by showing that \(\mathsf{C}(\eta)\) satisfies Condition~\eqref{C2}.
  Suppose that \(\mathsf{C}(0)\) has a model \(M\), then we have in particular \(M \models 0 \neq 0 + 0 = 0\).
  This is a contradiction, and therefore \(\mathsf{C}(0) \models \bot\).
  
  For Condition~\eqref{C1}, let \(M\) be a model of \(\mathsf{C}(s(\eta))\).
  Clearly, we have \(M \models \cnf(B + \axiomB{2})\), hence we only have to show that \(M \models \eta \neq x + x\) and \(M \models \eta \neq s(x + x)\).
  Suppose that \(M \models \eta = d + d\) for some \(d \in M\), then we have \(M \models s(\eta) = s(d + d)\).
  Since \(M \models \mathsf{C}(s(\eta))\), we also have \(M \models s(\eta) \neq s(d + d)\), a contradiction.
  Now suppose that \(M \models \eta = s(d + d)\) for some \(d \in M\).
  Since \(M \models \mathsf{C}(s(\eta))\), we also have \(M \models s(\eta) = s(s(d + d)) = s(d) + s(d)\).
  Thus \(M \models \mathsf{C}(\eta)\), that is, \(\mathsf{C}(s(\eta)) \models \mathsf{C}(\eta)\).
  
  Hence, \(\mathsf{C}(\eta)\) is a clause set cycle and therefore refutes itself.
\end{example}
The induction argument contained in a refutation by a clause set cycle is peculiar in the sense that it does not take place in an explicit background theory.
Instead of a background theory clause set cycles may contain clauses free of \(\eta\) that act as a background theory.
In the example above the clause set cycle contains the clauses \(\cnf(B + \axiomB{2})\), that correspond to the background theory.

The cycles detected by practical methods such as the n-clause calculus differ from clause set cycles in that they can be controlled by three parameters: An external offset, an internal offset, and the step size of the descent.
In the following we will show that these parameters do not increase the overall strength of the system.
\begin{definition}
  \label{def:2}
  Let \(L\) be a first-order language and \(j, k \in \mathbb{N}\) with \(j \geq 1\).
  A finite \(L \cup \{ \eta \}\) clause set \(\mathcal{C}(\eta)\) is called an \(L\) \((j,k)\)-clause set cycle if
  \begin{align*}
    \mathcal{C}(s^{j+k}(\eta))
    & \models \mathcal{C}(s^{k}(\eta)), \tag{C1'} \label{RC4'} \\
      \mathcal{C}(\numeral{m + k})
    & \models \bot, \ \text{for \(m = 0, \dots, j - 1\)}. \tag{C2'} \label{RC3'}
  \end{align*}
  We call the parameter \(j\) the descent step size and \(k\) the internal offset.
  Let \(i \in \mathbb{N}\) and \(\mathcal{D}(\eta)\) an \(L \cup \{\eta\}\) clause set, then \(\mathcal{D}(\eta)\) is refuted by the \((j, k)\)-clause set cycle \(\mathcal{C}(\eta)\) with external offset \(i\), if
  \begin{align*}
      \mathcal{D}(s^{i}(\eta))
      & \models \mathcal{C}(s^{k}(\eta)), \tag{C3'} \label{RC2'} \\
      \mathcal{D}(\numeral{m})
      & \models \bot, \ \text{for \(m = 0, \dots, i - 1\)}. \tag{C3''} \label{RC1'}
  \end{align*}
\end{definition}
Clearly, clause set cycles in the sense of Definition~\ref{def:1} are exactly the \((1,0)\)-clause set cycles and a refutation by a clause set cycle in the sense of Definition~\ref{def:1} is a refutation by a \((1,0)\)-clause set cycle with external offset~\(0\).

We start by showing that \((j,k)\)-clause set cycles with \(j, k \in \mathbb{N}\) and \(j \geq 1\) can be simulated by clause set cycles.
\begin{lemma}
  \label{lem:csc_with_parameters_are_simulated_by_csc}
  Let \(L\) be a first-order language, \(j, k \in \mathbb{N}\) with \(j \geq 1\), and \(\mathcal{C}(\eta)\) an \(L\) \((j,k)\)-clause set cycle.
  Then there exists a clause set cycle \(\mathcal{C}'(\eta)\) such that \(\mathcal{C}(s^{k}(\eta)) \models \mathcal{C}'(\eta)\).
\end{lemma}
\begin{proof}
  We start by eliminating the internal offset of the \((j,k)\)-clause set cycle, by letting \(\mathcal{C}'(\eta) \coloneqq \mathcal{C}(s^{k}(\eta))\).
  It is clear that \(\mathcal{C}'\) is a \((j,0)\)-clause set cycle.
  Moreover by the definition of \(\mathcal{C}'\) we have \(\mathcal{C}(s^{k}(\eta)) \models \mathcal{C}'(\eta)\).
  Let \(\mathcal{C}^{\prime\prime}(\eta)\) be the clause set obtained by applying Lemma~\ref{lem:disjunction_of_clause_sets} to the set \(\mathfrak{C} \coloneqq \{\mathcal{C}'(s^{m}(\eta)) \mid m = 0, \dots, j - 1\}\).
  We will now show that \(\mathcal{C}^{\prime\prime}(\eta)\) is a clause set cycle.
  Suppose that \(M \models \mathcal{C}^{\prime\prime}(0)\), then \(M \models \mathcal{C}^{\prime}(\numeral{m})\) for some \(m \in \{ 0, \dots, j - 1\}\), which is impossible and therefore \(\mathcal{C}^{\prime\prime}(0) \models \bot\).
  Now suppose that \(M \models \mathcal{C}^{\prime\prime}(s(\eta))\).
  Then we have \(M \models \mathcal{C}^{\prime}(s^{m + 1}(\eta))\) for some \(m \in \{0, \dots, j -1\}\).
  If \(m+1 \leq j - 1\), then \(\mathcal{C}(s^{m + 1}(\eta)) \in \mathfrak{C}\) and therefore \(M \models \mathcal{C}^{\prime\prime}(\eta)\).
  Otherwise we have \(m + 1 = j\) and therefore by \ref{RC4'} we obtain  \(M \models \mathcal{C}'(\eta)\) and since \(\mathcal{C}^{\prime}(\eta) \in \mathfrak{C}\), we have \(\mathcal{C}^{\prime}(\eta) \models \mathcal{C}^{\prime\prime}(\eta)\).
\end{proof}
Now we can show that a refutation by a \((j,k)\)-clause set cycle with internal offset \(i\), where \(i, j, k \in \mathbb{N}\) with \(j \geq 1\) can be reduced to a refutation by a clause set cycle.
\begin{proposition}
  \label{pro:7}
  \label{pro:refutation_by_csc_with_parameters_is_not_more_powerful}
  Let \(L\) be a first-order language, \(\mathcal{D}(\eta)\) an \(L \cup \{\eta\}\) clause set, and \(i,j, k \in \mathbb{N}\) with \(j \geq 1\) such that \(\mathcal{D}\) is refuted by an \(L\) \((j,k)\)-clause set cycle with external offset \(i\).
  Then \(\mathcal{D}(\eta)\) is refuted by a clause set cycle.
\end{proposition}
\begin{proof}
  Let \(\mathcal{C}(\eta)\) be an \(L\) \((j,k)\)-clause set cycle such that \(\mathcal{D}(\eta)\) is refuted by \(\mathcal{C}\) with external offset \(i\).
  By Lemma \ref{lem:csc_with_parameters_are_simulated_by_csc} there exists a clause set cycle \(\mathcal{C}^{\prime}(\eta)\) such that \(\mathcal{C}(s^{k}(\eta)) \models \mathcal{C}^{\prime}(\eta)\).
  Hence \(\mathcal{D}\) is refuted by a \((1,0)\)-clause set cycle with external offset \(i\).
  In the next step we will eliminate the external offset.
  Let \(\mathfrak{C} \coloneqq \{ \mathcal{D}(s^{m}(\eta)) \mid m = 0, \dots, i - 1\} \cup \{ \mathcal{C}'(\eta)\}\) and apply Lemma \ref{lem:disjunction_of_clause_sets} in order to obtain a clause set \(\mathcal{C}^{\prime\prime}(\eta)\) corresponding to the disjunction of the clause sets in \(\mathfrak{C}\).
  We will now show that \(\mathcal{C}^{\prime\prime}(\eta)\) is a clause set cycle.
  Suppose that \(M \models \mathcal{C}^{\prime\prime}(0)\), then either \(M \models \mathcal{D}(s^{m}(\eta))\) for some \(m \in \{ 0, \dots, i - 1\}\) or \(M \models \mathcal{C}^{\prime}(0)\).
  The first case is impossible because of Condition~\ref{RC1'} and the second case is impossible because \(\mathcal{C}'(\eta)\) is a clause set cycle and therefore \(\mathcal{C}'(0) \models \bot\).
  Hence we have \(\mathcal{C}^{\prime\prime}(0) \models \bot\).
  Now suppose that \(M \models \mathcal{C}^{\prime\prime}(s(\eta))\).
  If \(M \models \mathcal{C}^{\prime}(s(\eta))\), then we have \(M \models \mathcal{C}^{\prime}(\eta)\) because \(\mathcal{C}^{\prime}\) is a clause set cycle and therefore \(M \models \mathcal{C}^{\prime\prime}(\eta)\).
  If \(M \models \mathcal{D}(s^{m + 1}(\eta))\) for some \(m \in \{ 0, \dots, i - 1\}\), we need to consider two cases.
  If \(m + 1 < i\), then we have \(\mathcal{D}(s^{m+1}(\eta)) \in \mathfrak{C}\) and therefore \(M \models \mathcal{C}^{\prime\prime}(\eta)\).
  Otherwise we have \(m + 1 = i\), and therefore we obtain \(M \models \mathcal{C}'(\eta)\) by Condition~\ref{RC2'}.
  Again we obtain \(M \models \mathcal{C}^{\prime\prime}(\eta)\).
  Hence \(\mathcal{C}^{\prime\prime}(\eta)\) is a clause set cycle.
  We complete the proof by observing that \(\mathcal{D}(\eta) \models \mathcal{C}^{\prime\prime}(\eta)\), since \(\mathcal{D}(\eta) \in \mathfrak{C}\).
  Hence, \(\mathcal{D}(\eta)\) is refuted by the clause set cycle \(\mathcal{C}^{\prime\prime}(\eta)\).
\end{proof}
As already mentioned earlier, the notion of refutation by a clause set cycle is a useful intermediary abstraction of the induction mechanism of a family of \ac{AITP} systems including in particular the n-clause calculus \cite{kersani2013,kersani2014}.
Since our goal is to develop a uniform logical representation of methods for \ac{AITP}, we thus use the notion of refutation by a clause set cycle as a starting point to provide logical abstractions of \ac{AITP} systems such as the n-clause calculus.
In particular, we want, for a fixed language \(L\), to provide a logical \(L \cup \{ \eta \}\) theory \(T\)
that simulates refutation by a clause set cycle in the following sense: Let \(\mathcal{D}(\eta)\) be an \(L \cup \{\eta\}\) clause set that is refuted by an \(L\) clause set cycle, then \(T + \mathcal{D}(\eta)\) is inconsistent.
The authors of this article have shown in \cite{hetzl2020} that refutation by \(L\) clause set cycles can be simulated by the theory \(\IND{\exists_{1}(L)}\) (see Theorem \ref{thm:4}) and moreover that \(\IND{\Open(L)}\) does not simulate refutations by a clause set cycle.
\begin{theorem}[{\cite[Theorem~4.6]{hetzl2020}}]
  \label{thm:ref_csc_lower_bound_open_induction}
  There exists a language \(L\) and an \(L \cup \{\eta\}\) clause set \(\mathcal{D}(\eta)\) such that \(\mathcal{D}(\eta)\) is refuted by an \(L\) clause set cycle, but \(\IND{\Open(L)} + \mathcal{D}(\eta)\) is consistent.
\end{theorem}
In the following section, we will give a proof of Theorem~\ref{thm:ref_csc_lower_bound_open_induction} that is simpler, shorter, and more elegant than the proof given in \cite{hetzl2020}.
\begin{definition}
  \label{def:truncated_subtraction}
  Let \(n,m\) be natural numbers, then by \(n \dotminus m\) we denote the truncated subtraction of \(m\) from \(n\) given by
  \[
    n \dotminus m \coloneqq
    \begin{cases}
      n - m & \text{if \(n \geq m\)}
      \\
      0 & \text{otherwise}
    \end{cases}.
  \]
\end{definition}
\begin{lemma}
  \label{lem:5}
  \(\BLA + \IND{\Open(\LLA)} \not \vdash \Exists{y}{( x = y + y \vee x = s(y + y))}\)
\end{lemma}
\begin{proof}
  By Theorem~\ref{thm:shoenfields_theorem} it suffices to show that \(\BLAPrime \not \vdash \Exists{y}{( x = y + y \vee x = s( y + y ))}\).
  Consider the \(\LLA\) structure \(M\) whose domain consists of the pairs of the form \((m,n) \in \mathbb{N} \times \mathbb{Z}\) such that \(m = 0\) implies \(n \in \mathbb{N}\) and that interprets the non-logical symbols as follows:
  \begin{gather*}
    0^{M} = (0,0), \\
    s^{M}((m,n)) = (m, n + 1), \\
    p^{M}((m,n)) =
    \begin{cases}
      (m,n \dotminus 1) & \text{if \(m = 0\)}, \\
      (m,n - 1) & \text{otherwise}
    \end{cases}, \\
    (m_{1}, n_{1}) +^{M} (m_{2},n_{2}) = (m_{1} + m_{2}, n_{1} + n_{2}).
  \end{gather*}
  It is routine to verify that \(M \models \BLAPrime\).
  Consider the element \((1,0)\), then clearly there is no element \((m,n)\) of \(M\) such that \((1,0) = (m,n) +^{M} (m,n) = (2m,2n)\) or \((1,0) = s^{M}((m,n) +^{M} (m,n)) = (2m,2n+1)\).
\end{proof}
\begin{proof}[Proof of Theorem~\ref{thm:ref_csc_lower_bound_open_induction}]
  Consider the clause set \(\mathsf{C}(\eta)\).
  In Example~\ref{ex:1} we have shown that \(\mathsf{C}(\eta)\) is refuted by an \(\LLA\) clause set cycle.
  We will now show that \(\IND{\Open(\LLA)} + \mathsf{C}(\eta)\) is consistent.
  We proceed indirectly and assume that \(\IND{\Open(\LLA)} + \mathsf{C}(\eta)\) is inconsistent.
  Hence \(\BLA + \axiomB{2} + \IND{\Open(\LLA)} \vdash \Exists{y}{( \eta = y + y )} \vee \Exists{y}{( \eta = s( y + y))}\).
  Thus, \(\BLA + \IND{\Open(\LLA)} \vdash \Exists{y}{(x = y + y \vee x = s(y + y))}\), which contradicts Lemma~\ref{lem:5}.
\end{proof}

However, empirical evidence suggests that clause set cycles are not strictly stronger than open induction.
This has given rise to the following conjecture.
\begin{conjecture}[{\cite[Conjecture~4.7]{hetzl2020}}]
  \label{con:csc_and_open_induction_incomparable}
  \label{con:1}
  There exists a language \(L\) and an \(L \cup \{ \eta \}\) clause set \(\mathcal{D}(\eta)\) such that \(\IND{\Open(L)} \cup \mathcal{D}(\eta)\) is inconsistent, but \(\mathcal{D}(\eta)\) is not refuted by an \(L\) clause set cycle.
\end{conjecture}
In the following section we will give a characterization of refutation by a clause set cycle in terms of a logical theory. In Section~\ref{sec:unpr-clause-set} we will make use of this characterization to give a positive answer to Conjecture~\ref{con:csc_and_open_induction_incomparable}.
\subsection{Logical characterization}
\label{sec:refin-bound-theory}
In the previous section we have introduced the notion of refutation by a clause set cycle and we have shown that certain practically motivated generalizations of refutation by a clause set cycle do not result in stronger systems.
In this section we will give a characterization of refutation by a clause set cycle in terms of a logical theory.

We start by converting clause set cycles into formulas.  
\begin{lemma}
  \label{lem:7}
  Let \(\mathcal{C}(\eta)\) be an \(L\) clause set cycle, then the formula \(\neg \cnf^{-1}(\mathcal{C})[\eta/x]\) is \(\varnothing\)-inductive.
  Let \(\mathcal{D}(\eta)\) be an \(L \cup \{ \eta \}\) clause set that is refuted by the clause set cycle \(\mathcal{C}(\eta)\), then \(\neg \cnf^{-1}(\mathcal{C}) + \mathcal{D}(\eta)\) is inconsistent.
\end{lemma}
\begin{proof}
  Clearly, we have \(M \models \neg \cnf^{-1}(\mathcal{C})\) if and only if \(M \not \models \mathcal{C}\).
  Hence, \(\models \neg \cnf^{-1}(\mathcal{C}(0))\) and \(\neg \cnf^{-1}(\mathcal{C}(\eta)) \models \neg \cnf^{-1}(\mathcal{C}(s(\eta)))\).
  Therefore, by the completeness theorem and the deduction theorem for first-order logic we have
  \begin{gather*}
    \vdash \neg \cnf^{-1}(\mathcal{C}(0)), \\
    \vdash \neg \cnf^{-1}(\mathcal{C}(\eta)) \rightarrow \neg \cnf^{-1}(\mathcal{C}(s(\eta))).
  \end{gather*}
  Thus, \(\vdash \neg \cnf^{-1}(\mathcal{C})[\eta/0]\) and  \(\vdash \neg \cnf^{-1}(\mathcal{C})[\eta/x] \rightarrow \neg \cnf^{-1}(\mathcal{C})[\eta/s(x)]\).
  The second part of the lemma is obvious.
\end{proof}
Let \(\mathcal{C}\) be a clause set cycle, then the formula \(\neg \cnf^{-1}(\mathcal{C})[\eta/x]\) is the formula that corresponds to the induction argument contained in a refutation by a clause set cycle.
Clearly, this formula is logically equivalent to an \(\exists_{1}\) formula.
In the following we will make three further important observations about this argument by induction.

The first observation is that the formula \(\neg \cnf^{-1}(\mathcal{C})[\eta/x]\) has only one free variable, that is, the variable on which the argument by induction takes places. 
Hence the induction captured by clause set cycles is essentially parameter-free induction.
In this article we use a notation for parameter-free induction that is inspired by the notation used in the literature from mathematical logic on parameter-free induction \cite{adamowicz1987,kaye1988,beklemishev1997b,beklemishev1999,cordon2011}.
\begin{definition}
  \label{def:parameter_free_induction_schema}
  Let \(\Gamma\) be a set of formulas, then \(\INDParameterFree{\Gamma}\) is axiomatized by the universal closure of the formulas \(I_{x}\varphi\) for \(\varphi(x) \in \Gamma\).
\end{definition}
When the set of induction formulas is unrestricted, induction without parameters is just as powerful as induction with parameters.
\begin{lemma}
  \label{lem:16}
  Let \(L\) be a first-order language, then we have
  \[
    \IND{\mathcal{F}(L)} \equiv \INDParameterFree{\mathcal{F}(L)}.
  \]
\end{lemma}
\begin{proof}
  We only show \(\INDParameterFree{\mathcal{F}(L)} \vdash \IND{\mathcal{F}(L)}\), the other direction is trivial.
  Let \(\varphi(x,\vec{z})\) be an \(L\) formula, \(x\) a variable, and \(\vec{z}\) a vector of variables.
  We let the formula \(\psi(x)\) be given by \[
    \Forall{\vec{z}}{(\varphi(0,\vec{z}) \wedge \Forall{x}{(\varphi(x, \vec{z}) \rightarrow \varphi(s(x), \vec{z}))} \rightarrow \varphi(x, \vec{z}))}.
  \]
  By a straightforward quantifier shift we obtain \(\vdash \Forall{x}{\psi(x)} \leftrightarrow \Forall{\vec{z}}{I_{x}\varphi(x,\vec{z})}\).
  Furthermore, it is straightforward to check that \(\vdash \psi(0)\) and \(\vdash \psi(x) \rightarrow \psi(s(x))\).
  Hence, \(\vdash I_{x}\psi \rightarrow \Forall{x}{\psi}\).
  Therefore \(\vdash I_{x}\psi \rightarrow \Forall{\vec{z}}{I_{x}\varphi}\).
\end{proof}
However, when we are dealing with restricted induction schemes such as \(\IND{\exists_{k}(L)}\), then its parameter-free counterpart \(\INDParameterFree{\exists_{k}(L)}\) may be a weaker theory \cite{kaye1988}.

Another remarkable property of the formula \(\neg \cnf^{-1}(\mathcal{C})[\eta/x]\) is its \(\varnothing\)-inductivity.
In a refutation by a clause set cycle, there is no explicit induction axiom. Instead, whenever a clause set \(\mathcal{C}(\eta)\) is shown to be a clause set cycle, it can be used in a refutation.
This is reminiscent of a Hilbert-style induction rule that allows us to deduce \(\cnf^{-1}(\mathcal{C})[\eta/x]\) if \(\cnf^{-1}(\mathcal{C})[\eta/x]\) is \(\varnothing\)-inductive.
The idea of Hilbert-style inference rules and in particular of induction rules is made explicit in the following two definitions.
\begin{definition}
  \label{def:6}
  An inference rule \(R\) is a set of tuples of the form \(\Gamma/\gamma_{0}\) called the instances of \(R\), where \(\Gamma = \{ \gamma_{1}, \dots, \gamma_{n}\}\) is a finite set of sentences and \(\gamma_{0}\) is a sentence.
  Let \(T\) be a theory, then the theory of unnested applications \([T,R]\) of the inference rule \(R\) over the theory \(T\) is axiomatized by
  \[
    T + \{ \varphi \mid T \vdash \Gamma, \Gamma/\varphi \in R\}.
  \]
  Let \([T, R]_{0} \coloneqq T\) and \([T,R]_{n + 1} = [[T,R]_{n},R]\), then we define \(T + R \coloneqq \bigcup_{n \geq 0}[T,R]_{n}\).
\end{definition}
Let \(R\) be an inference rule and \(\Gamma/\gamma_{0} \in R\), then the intended meaning of the rule instance \(\Gamma/\gamma_{0}\) is that whenever all the sentences in \(\Gamma\) are derived, then we can derive \(\gamma_{0}\).
The instance \(\Gamma/\gamma_{0}\) will also be written as
\[
  \begin{prooftree}
    \hypo{\gamma_{1}}
    \hypo{\dots}
    \hypo{\gamma_{n}}
    \infer3{\gamma_{0}}
  \end{prooftree}
\]
\begin{definition}
  \label{def:16}
  Let \(\Gamma\) be a set of formulas, then the rule \(\RIND{\Gamma}\) consists of the instances of the form
  \begin{equation*}
    \begin{prooftree}
      \hypo{\Forall{\vec{z}}{\gamma(0,\vec{z})}}
      \hypo{\Forall{\vec{z}}{\Forall{x}{(\gamma(x,\vec{z}) \rightarrow \gamma(s(x),\vec{z}))}}}
      \infer2{\Forall{\vec{z}}{\Forall{x}{\gamma(x,\vec{z})}}}
    \end{prooftree},
  \end{equation*}
  with \(\gamma \in \Gamma\) and where the variable \(x\) is called the induction variable and the variables \(\vec{z}\) are called the induction parameters.
  The induction rule \(\RINDParameterFree{\Gamma}\) consists of these instances of \(\RIND{\Gamma}\) where the induction variable is the only free variable of the induction formula.
\end{definition}
Let \(T\) be a theory and \(\Gamma\) a set of formulas, then we can make use of Definition~\ref{def:22} to reformulate the theory \([T, \RIND{\Gamma}]\) as follows
\[
  [T, \RIND{\Gamma}] \equiv T + \{ \varphi \mid \varphi(x,\vec{z}) \in \Gamma, \text{\(\varphi\) is \(T\)-inductive in \(x\)}\}.
\]
In other words the theory \([T,\RIND{\Gamma}]\) provides induction only for \(T\)-inductive formulas from \(\Gamma\), whereas \(T + \IND{\Gamma}\) provides induction for all formulas in \(\Gamma\).
It is obvious that \(T + \IND{\Gamma} \vdash [T, \RIND{\Gamma}]\).
However, \([T, \RIND{\Gamma}]\) is in general not as strong as \(T + \IND{\Gamma}\), see~\cite{parsons1972}.
For further literature on induction rules, see for example \cite{shoenfield1958,shepherdson1963,parsons1972,beklemishev1997,jerabek2020}.

We will now make a last observation about the argument by induction contained in a refutation by a clause set cycle.
The previous observations show that clause set cycles are simulated by unnested applications of the parameter-free \(\exists_{1}\) induction rule over the theory \(\varnothing\).
A sentence derived by an induction rule is the universal closure of an inductive formula.
Hence, once a formula is derived by an induction rule it can be instantiated freely.
Similarly, a clause set cycle \(\mathcal{C}(\eta)\) acts, roughly speaking, as the lemma \(\neg \cnf^{-1}(\mathcal{C})[\eta/x]\) of which, however, only the instance \(\neg \cnf^{-1}(\mathcal{C})\) is used.
In other words, a clause set cycle allows us to derive properties of \(\eta\) only.
We will informally refer to this restriction as the instance restriction.
We can capture this restriction in the following restricted induction rule.
\begin{definition}
  \label{def:23}
  Let \(\Gamma\) be a set of formulas, then the rule \(\RIND{\Gamma}_{\eta}\) consists of the instances of the form
    \begin{equation*}
    \begin{prooftree}
      \hypo{\Forall{\vec{z}}{\gamma(0,\vec{z})}}
      \hypo{\Forall{\vec{z}}{\Forall{x}{(\gamma(x,\vec{z}) \rightarrow \gamma(s(x),\vec{z}))}}}
      \infer2{\Forall{\vec{z}}{\gamma(\eta,\vec{z})}}
    \end{prooftree}, \ \text{with \(\gamma \in \Gamma\).}
  \end{equation*}
  The rule \(\RINDParameterFreeEta{\Gamma}\) consists of those instances of \(\RIND{\Gamma}_{\eta}\) where the induction variable is the only free variable of the induction formula.
\end{definition}
By combining the above observations we obtain the following proposition, that allows us to simulate clause set cycles in a logical theory.
\begin{proposition}
  \label{pro:9}
    Let \(\mathcal{D}(\eta)\) be an \(L \cup \{ \eta \}\) clause set.
  If \(\mathcal{D}(\eta)\) is refuted by an \(L\) clause set cycle, then \([\varnothing, \RINDParameterFreeEta{\exists_{1}(L)}] \cup \mathcal{D}(\eta)\) is inconsistent.
\end{proposition}
\begin{proof}
  Since \(\mathcal{D}\) is refuted by a clause set cycle, there exists an \(L\) clause set cycle \(\mathcal{C}(\eta)\) such that
  \begin{equation}
    \label{eq:3}
    \mathcal{D}(\eta) \models \mathcal{C}(\eta). \tag{*}
  \end{equation}
  Let \(\varphi(x) \coloneqq \cnf^{-1}(\mathcal{D})[\eta/x]\) then \(\varphi(\eta)\) is clearly logically equivalent to \(\mathcal{D}(\eta)\).
  By the soundness of first-order logic it thus suffices to show that
  \[
    [\varnothing, \RINDParameterFree{\exists_{1}(L)}] \vdash \neg \varphi(\eta).
  \]
  Let \(\psi(x)\) be an \(\exists_{1}\) formula that is logically equivalent to \(\neg \cnf^{-1}(\mathcal{C})[\eta/x]\).
  Then, by applying the completeness theorem and the deduction theorem to \eqref{eq:3}, we obtain
  \begin{equation}
    \label{eq:7}
    \tag{\(\dagger\)}
    \vdash \varphi(\eta) \rightarrow \neg \psi(\eta).
  \end{equation}
  By Lemma \ref{lem:7} we know that \(\psi(x)\) is \(\varnothing\)-inductive, and therefore we have \([\varnothing, \RINDParameterFreeEta{\exists_{1}(L)}] \vdash \psi(\eta)\).
  Hence, by considering the contrapositive of \eqref{eq:7} we clearly obtain \([\varnothing, \RINDParameterFreeEta{\exists_{1}(L)}] \vdash \neg \varphi(\eta)\).
\end{proof}
We will now show that we even have the converse and thus obtain a characterization of refutation by a clause set cycle by a logical theory.
We start by observing that finitely many inductive formulas can be fused into a single inductive formula.
\begin{lemma}
  \label{lem:18}
  Let \(T\) be a theory and let \(\varphi_{1}(x, \vec{z})\), \dots, \(\varphi_{n}(x, \vec{z})\) be formulas.
  If \(\varphi_{i}\) is  \(T\)-inductive in \(x\) for \(i = 1, \dots, n\), then \(\psi \coloneqq \bigwedge_{i = 1, \dots n}\varphi_{i}\) is \(T\)-inductive in \(x\).
\end{lemma}
\begin{proof}
  We start by showing that \(T \vdash \psi(0,\vec{z})\).
  Let \(j \in \{ 1, \dots, n\}\), then since \(\varphi_{j}\) is \(T\)-inductive in \(x\), we have \(T \vdash \varphi_{j}(0, \vec{z})\) and we are done.
  Now let us show that \(T \vdash \psi(x,\vec{z}) \rightarrow \psi(s(x),\vec{z})\).
  Work in \(T\), assume \(\bigwedge_{i = 1}^{n}\varphi_{i}(x,\vec{z})\), and let \(j \in \{ 1, \dots, n\}\).
  Since \(\varphi_{j}\) is \(T\)-inductive in \(x\), we have \(\varphi_{j}(x,\vec{z}) \rightarrow \varphi_{j}(s(x),\vec{z})\).
  Hence we obtain \(\varphi_{j}(s(x),\vec{z})\) and therefore \(\psi\) is \(T\)-inductive in \(x\).
\end{proof}
This simple result is particularly interesting because fusing inductive formulas neither introduces more induction parameters and when fusing \(\exists_{k}\) induction formulas, the fused induction formula is also logically equivalent to an \(\exists_{k}\) formula.
Similar techniques exist for fusing a finite number of induction axioms into a single induction axiom \cite{wong2018,gentzen1954}.
However, these either introduce a new induction parameter or increase the quantifier complexity of the resulting induction formula.
\begin{proposition}
  \label{pro:8}
  Let \(\mathcal{D}(\eta)\) be an \(L \cup \{ \eta \}\) clause set. If \([\varnothing, \RINDParameterFreeEta{\exists_{1}(L)}] + \mathcal{D}(\eta)\) is inconsistent, then \(\mathcal{D}(\eta)\) is refuted by an \(L\) clause set cycle.
\end{proposition}
\begin{proof}
  Let \(\varphi(x) \coloneqq \cnf^{-1}(\mathcal{D})[\eta/x]\), then by the completeness theorem and the deduction theorem we obtain \([\varnothing, \RINDParameterFreeEta{\exists_{1}(L)}] \vdash \neg \varphi(\eta)\).
  By the compactness theorem there exist \(\exists_{1}\) \(L\) formulas \(\psi_{1}(x)\), \dots, \(\psi_{k}(x)\) such that \(\psi_{i}\) is \(\varnothing\)-inductive for \(i = 1, \dots, k\) and
  \[
    \psi_{1}(\eta) + \dots + \psi_{k}(\eta) \vdash \neg \varphi(\eta).
  \]
  By Lemma~\ref{lem:18}, the formula \(\Psi(x) \coloneqq \bigwedge_{i = 1}^{k}\psi_{i}\) is \(\varnothing\)-inductive.
  Moreover we have \(\Psi(\eta) \vdash \neg \varphi(\eta)\).
  Clearly, \(\Psi\) is logically equivalent to an \(\exists_{1}\) formula, hence there exists a \(\forall_{1}\) formula \(\Theta\) that is logically equivalent to \(\neg \Psi\).
  Since \(\vdash \Psi(0)\) and \(\vdash \Psi(x) \rightarrow \Psi(s(x))\), we have \(\Theta(0) \models \bot\) and \(\Theta(s(x)) \models \Theta(x)\).
  Therefore, \(\mathcal{C} \coloneqq \cnf(\Theta(\eta))\) is a clause set cycle.
  Finally, since \(\Psi(\eta) \vdash \neg \varphi(\eta)\), we obtain \(\varphi(\eta) \models \neg \Psi(\eta)\), that is, \(\mathcal{D}(\eta) \models \mathcal{C}(\eta)\).
  In other words, \(\mathcal{D}\) is refuted by the clause set cycle \(\mathcal{C}\).
\end{proof}
We thus obtain a characterization of  refutation by a clause set cycle in terms of induction rules.
\begin{theorem}
  \label{thm:7}
  \label{thm:ref_csc_logical_characterization_clause_sets}
  Let \(\mathcal{D}(\eta)\) be an \(L \cup \{ \eta \}\) clause set, then \(\mathcal{D}\) is refuted by an \(L\) clause set cycle if and only if \([\varnothing,\RINDParameterFreeEta{\exists_{1}(L)}] + \mathcal{D}(\eta)\) is inconsistent.
\end{theorem}
\begin{proof}
  An immediate consequence of Propositions \ref{pro:9} and \ref{pro:8}.
\end{proof}
\begin{remark}
  \label{rem:3}
  In a refutation by a clause set cycle the constant \(\eta\) plays essentially two roles: On the one hand, it can be thought of as a Skolem symbol and, on the other hand, it plays the role of an induction variable.
  The characterization of Theorem~\ref{thm:ref_csc_logical_characterization_clause_sets} clarifies this situation by allowing us to distinguish between induction variables and the Skolem symbol \(\eta\).
\end{remark}
As a corollary we obtain Theorem~2.10 of \cite{hetzl2020}.
\begin{theorem}[{\cite[Theorem~2.10]{hetzl2020}}]
  \label{thm:4}
  Let \(L\) be a first-order language and \(\mathcal{D}(\eta)\) an \(L \cup \{\eta\}\) clause set.
  If \(\mathcal{D}(\eta)\) is refuted by an \(L\) clause set cycle, then \(\IND{\exists_{1}(L)} + \mathcal{D}(\eta)\) is inconsistent.
\end{theorem}
\begin{proof}
  Obvious, since \(\IND{\exists_{1}(L)} \vdash [\varnothing, \RINDParameterFreeEta{\exists_{1}(L)}]\).
\end{proof}
In the following section we will make use of the characterization of Theorem~\ref{thm:ref_csc_logical_characterization_clause_sets} to construct clause sets that are refutable by open induction but which are not refutable by clause set cycles.
In particular the unrefutability results that we provide exploit different logical features of clause set cycles.
\subsection{Unprovability by clause set cycles}
\label{sec:unpr-clause-set}
In the previous sections we have introduced the notion of refutation by a clause set cycle for which we have shown a characterization in terms of a logical theory.
We have shown this characterization by discerning four main logical features of refutation by a clause set cycle: the quantifier-complexity, the absence of induction parameters, the similarity with induction rules, and the restriction on instances of derived formulas.
In this section we will make use of this characterization in order to provide practically relevant clause sets that are not refutable by clause set cycles, but that are refutable by induction on quantifier-free formulas.
The unrefutability results in this section will exploit different logical features of clause set cycles.
In particular we will show that restricting the instances of the conclusion of the induction rule can be very drastic.

Let us now briefly discuss the practical applicability of the unprovability results given in this section.
The unprovability results apply to any sound (for first-order logic) saturation prover that detects clause set cycles over the language of the initial clause set.
Hence, our unprovability results apply, in particular, to all sound saturation provers that do not extend the language of the initial clause set and detect cycles among the derived clauses such as for example the n-clause calculus (see \cite{kersani2013,kersani2014}).
On the other hand systems that extend the language are also of practical importance, since such extensions can be used to organize the refutation process, see for example \cite{voronkov2014}.
In particular, the extension of the language by definitions can be expected to have interesting effects.
However, investigating the interaction between clause set cycles and various language extending mechanisms would go beyond the scope of this article and should be investigated separately.
Observe, furthermore, that our setting does not rule out the presence of Skolem symbols other than \(\eta\) in clause set cycles.

We start by slightly reformulating Theorem~\ref{thm:ref_csc_logical_characterization_clause_sets} so that we can work with formulas and theories instead of clause sets.
\begin{corollary}
  \label{cor:1}
  Let \(L\) be a first-order language, \(T\) a \(\forall_{1}\) \(L\) theory, and \(\varphi(x,\vec{y})\) a quantifier-free \(L\) formula, then \(T + [\varnothing, \RINDParameterFreeEta{\exists_{1}(L)}] \vdash \Exists{\vec{y}}{\varphi(\eta, \vec{y})}\) if and only if \(\cnf(T + \Forall{\vec{y}}{\neg \varphi(\eta, \vec{y})})\) is refuted by an \(L\) clause set cycle.
\end{corollary}
\begin{proof}
  Clearly, \(T + [\varnothing, \RINDParameterFreeEta{\exists_{1}(L)}] \vdash \Exists{\vec{y}}{\varphi(\eta, \vec{y})}\) if and only if \([\varnothing, \RINDParameterFreeEta{\exists_{1}(L)}] + \cnf(T + \Forall{\vec{y}}{\neg \varphi(\eta, \vec{y})})\) is inconsistent.
  By Theorem~\ref{thm:ref_csc_logical_characterization_clause_sets}, \([\varnothing, \RINDParameterFreeEta{\exists_{1}(L)}] + \cnf(T + \Forall{\vec{y}}{\neg \varphi(\eta, \vec{y})})\) is inconsistent if and only if \(\cnf(T + \Forall{\vec{y}}{\neg \varphi(\eta, \vec{y})})\) is refuted by an \(L\) clause set cycle.
\end{proof}
In Section~\ref{sec:definition-bound} we have informally observed that clause set cycles do not take place in some explicit background theory but instead clause set cycles contain the clauses corresponding to the background theory.
In the following we will make this informal observation more precise.
\begin{lemma}
  \label{lem:10}
  Let \(L\) be a first-order language, \(T\) a \(\forall_{1}\) \(L\) theory, \(U\) an \(L\) theory, then
  \[
    T + [U, \RINDParameterFreeEta{\exists_{1}(L)}] \equiv [T + U, \RINDParameterFreeEta{\exists_{1}(L)}].
  \]
  Furthermore, \(T + [U, \RINDParameterFree{\exists_{1}(L)}] \equiv [T + U, \RINDParameterFree{\exists_{1}(L)}]\).
\end{lemma}
\begin{proof}
  The direction \([T + U, \RINDParameterFreeEta{\exists_{1}(L)}] \vdash T + [U, \RINDParameterFreeEta{\exists_{1}(L)}]\) is immediate.
  For the other direction let \(\gamma(x)\) be an \(\exists_{1}\) \(L\) formula and assume that \(T + U \vdash \gamma(0)\) and \(T + U \vdash \gamma(x) \rightarrow \gamma(s(x))\).
  By the compactness theorem and the deduction theorem there exist \(\tau, \tau_{1}, \dots, \tau_{n} \in T\) such that \(\tau = \bigwedge_{i = 1}^{n}\tau_{i}\) and \(U \vdash \tau \rightarrow \gamma(0)\) and \(U \vdash \tau \rightarrow \gamma(x) \rightarrow \gamma(s(x))\).
  By straightforward propositional equivalences we obtain
  \[
    U \vdash \left(\tau \rightarrow \gamma(x)\right) \rightarrow \left(\tau \rightarrow \gamma(s(x))\right).
  \]
  Clearly, \(\tau\) is logically equivalent to a \(\forall_{1}\) sentence, hence \(\tau \rightarrow \gamma(x)\) is logically equivalent to an \(\exists_{1}\) formula \(\gamma'(x)\).
  Hence, \([U, \RINDParameterFreeEta{\exists_{1}(L)}] \vdash \gamma'(\eta)\) and therefore \([U, \RINDParameterFreeEta{\exists_{1}(L)}] \vdash \tau \rightarrow \gamma(\eta)\).
  Thus, \(T + [U, \RINDParameterFreeEta{\exists_{1}(L)}] \vdash \gamma(\eta)\).
  We show \(T + [U, \RINDParameterFree{\exists_{1}(L)}] \equiv [T + U, \RINDParameterFree{\exists_{1}(L)}]\) analogously, with the exception that in the last part of the argument we have to shift the universal quantifier in \(\Forall{x}{(\tau \rightarrow \gamma(x))}\) inwards.
\end{proof}
Lemma~\ref{lem:10} allows us to move \(\forall_{1}\) axioms in and out of the induction rule and thus to consider the \(\eta\)-free clauses of a clause set cycle as the background theory.
As an immediate consequence of Corollary~\ref{cor:1} and Lemma~\ref{lem:10} we now obtain a general pattern to reduce unrefutability problems for clause set cycles to independence problems.
\begin{proposition}
  \label{pro:10}
  Let \(L\) be a first-order language, \(T\) a \(\forall_{1}\) \(L\) theory, and let \(\varphi(x, \vec{y})\) be a quantifier-free \(L\) formula.
  Then \([T, \RINDParameterFreeEta{\exists_{1}(L)}] \vdash \Exists{\vec{y}}{\varphi(\eta, \vec{y})}\) if and only if the clause set \(\cnf(T + \Forall{\vec{y}}{\neg \varphi(\eta, \vec{y})})\) is refuted by an \(L\) clause set cycle.
\end{proposition}
\begin{proof}
  We have the following chain of equivalences:
  \begin{align*}
    [T,&\RINDParameterFreeEta{\exists_{1}(L)}] \vdash \Exists{\vec{y}}{\varphi(\eta, \vec{y})}, \\
    & \Leftrightarrow_{\text{Lem.~\ref{lem:10}}} T + [\varnothing, \RINDParameterFreeEta{\exists_{1}(L)}] \vdash \Exists{\vec{y}}{\varphi(\eta, \vec{y})} \\
    & \Leftrightarrow \text{\([\varnothing, \RINDParameterFreeEta{\exists_{1}(L)}] \cup \cnf(T + \Forall{\vec{y}}{\neg \varphi(\eta, \vec{y})})\) is inconsistent} \\
    & \Leftrightarrow_{\text{Cor.~\ref{cor:1}}} \text{\(\cnf(T + \Forall{\vec{y}}{\neg \varphi(\eta, \vec{y})})\) is refuted by an \(L\) clause set cycle.} \qedhere
  \end{align*}
\end{proof}
We can now consider some theories and formulas that will yield clause sets that are unrefutable by clause set cycles.
By the characterization of clause set cycles by a logical theory we have discerned several restrictions of the induction principle that corresponds to clause set cycles.
In the following two subsections we will formulate unprovability results that attack different restrictions of the induction principle that is contained in the notion of refutation by a clause set cycle.

\subsubsection{Instance restriction}
\label{sec:unrefutability_instance_restriction}
In Section~\ref{sec:refin-bound-theory} we have observed that a refutation by a clause set cycle only permits a single instance of a clause set cycle to appear in a refutation.
In this section we will formulate an unprovability result for clause set cycles that exploits this restriction.
In particular, we will base this unprovability result on a stronger independence result that shows how drastic the instance restriction is.
\begin{definition}
  \label{def:8}
  Let \(f/1\) be a function symbol and \(P/1\) be a predicate symbol.
  The theory \(\mathcal{P}\) is axiomatized by the universal closure of the following formulas
  \begin{gather*}
    0 \neq s(x), \\
    s(x) = s(y) \rightarrow x = y, \\
    P(0), \\
    P(x) \rightarrow P(s(x)).
  \end{gather*}
\end{definition}
\begin{definition}
  \label{def:25}
  Let \(\varphi(x, \vec{z})\) be a formula, then \(I_{x}^{\eta}\varphi\) denotes the formula
  \[
    \varphi(0, \vec{z}) \wedge \Forall{x}{(\varphi(x,\vec{z}) \rightarrow \varphi(s(x),\vec{z}))} \rightarrow \varphi(\eta, \vec{z}).
  \]
  Let \(\Gamma\) be a set of formulas, then the theory \(\IND{\Gamma}_{\eta}\) is axiomatized by the universal closure of the formulas \(I_{x}^{\eta}\gamma\) with \(\gamma \in \Gamma\).
\end{definition}
We have the following independence.
\begin{proposition}
  \label{pro:12}
  \(\mathcal{P} + \IND{\mathcal{F}(\{0, s, P, f\})}_{\eta} \not \vdash P(f(\eta))\).
\end{proposition}
\begin{proof}
  Let \(M\) be the \(\{0, s, P, f\}\) structure with domain consisting of pairs \((m,n) \in \{ 0 , 1 \} \times \mathbb{Z}\) such that if \(m = 0\), then \(n \in \mathbb{N}\).
  Let \(M\) interpret the non-logical symbols as follows
  \begin{gather*}
    0^{M} = \eta^{M} = (0,0), \\
    s^{M}((m,n)) = (m, n+1), \\
    f^{M}((m,n)) = (1, n), \\
    P^{M} = \{ (0,n) \mid n \in \mathbb{N} \}.
  \end{gather*}
  It is clear that \(M\) is a \(\{0, s, P, f\}\) structure and moreover it is straightforward to verify that \(M\) is a model of \(\mathcal{P}\).
  Now let us show that \(M \models \IND{\mathcal{F}(\{0, s, P, f\})}_{\eta}\).
  Let \(\psi(x, \vec{z})\) be a \(\{0, s, P, f\}\) formula, \(\vec{c}\) a vector of elements of \(M\).
  Assume that \(M \models \psi(0,\vec{c})\) and \(M \models \psi(x, \vec{c}) \rightarrow \psi(s(x), \vec{c})\).
  Since \(\eta^{M} = 0^{M}\), we already have \(M \models \psi(\eta, \vec{c})\) and therefore \(M \models I_{x}^{\eta}\psi(x, \vec{z})\).
  Finally, observe that \(f^{M}(\eta^{M}) = (1,0) \notin P^{M}\), hence \(\mathcal{P} + \IND{\mathcal{F}(\{0, s, P, f\})}_{\eta} \not \vdash P(f(\eta))\).
\end{proof}
The above independence result is remarkable in the sense that it imposes no restriction whatsoever on the induction formulas, only the conclusion of the induction axioms is restricted.
Hence the result shows that this restriction is extremely strong.
As a corollary we obtain the following unrefutability result for clause set cycles.
\begin{corollary}
  \label{cor:3}
  The \(\{0, s, P, f, \eta\}\) clause set \(\cnf(\mathcal{P} + \neg P(f(\eta)))\) is not refuted by a \(\{0, s, P, f\}\) clause set cycle.
\end{corollary}
\begin{proof}
  Suppose that \(\cnf(\mathcal{P} + \neg P(f(\eta)))\) is refuted by a \(\{0, s, P, f\}\) clause set cycle.
  Then, by Proposition~\ref{pro:10} we have \([\mathcal{P}, \RINDParameterFreeEta{\exists_{1}(\{0, s, P, f\})}] \vdash P(f(\eta))\).
  However, since \(\mathcal{P} + \IND{\mathcal{F}(\{0, s, P, f\})}_{\eta} \vdash [\mathcal{P}, \RINDParameterFreeEta{\exists_{1}(\{0, s, P, f\})}]\), this contradicts Proposition~\ref{pro:12}.
\end{proof}
\begin{lemma}
  \label{lem:6}
  \([\mathcal{P},\RINDParameterFree{\Open(\{0, s, P, f\})}] \vdash P(f(\eta))\).
\end{lemma}
\begin{proof}
  The formula \(P(x)\) is inductive in \(\mathcal{P}\).
\end{proof}
Proposition~\ref{pro:12}, Corollary~\ref{cor:3}, and Lemma~\ref{lem:6} together show that the \(\eta\)-restriction as encountered in the n-clause calculus is drastic and can result in pathological unrefutability phenomena.
On the one hand, without the \(\eta\)-restriction a very simple argument by induction suffices to prove \(P(f(\eta))\) and on the other hand in presence of the \(\eta\)-restriction even induction for all \(\{ 0, s, P, f \}\) formulas does not allow us to prove the formula \(P(f(\eta))\).
However, because of this the unrefutability result of Corollary~\ref{cor:3} does not tell us anything about the other restrictions of the induction principle contained in refutations by a clause set cycle.

Hence, it would be interesting to have a similar result for linear arithmetic.
In particular we conjecture the following.
\begin{conjecture}
  \([\BLA, \RINDParameterFreeEta{\exists_{1}(\LLA)}] \not \vdash 0 + (\eta + \eta) = (\eta + \eta)\).
\end{conjecture}

\subsubsection{Induction rule and absence of parameters}
\label{sec:unrefutability_rule_and_parameters}
In the following we will consider another unprovability result for clause set cycles that does not make use of the instance restriction, but instead exploits the absence of induction parameters and the induction rule.
This time we work in the setting of linear arithmetic described in Section~\ref{sec:linear-arithmetic}.
The unprovability result developed in this section is based on the following weak cancellation property of the addition of natural numbers.
\begin{definition}
  \label{def:14}
  Let \(k,n,m \in \mathbb{N}\) with \(0 < n < m\), then we define
  \begin{equation}
    \label{eq:4}
    n \cdot x + \numeral{(m - n)k} = m \cdot x \rightarrow x = \numeral{k}. \tag{\(\formulaE{k}{n}{m}\)}
  \end{equation}
\end{definition}
The formula $\formulaE{k}{n}{m}$ is a generalization of
\begin{equation}
 x + 0 = x + x \rightarrow x = 0. \tag{\(\formulaE{0}{1}{2}\)}
\end{equation}
Most of the upcoming Section~\ref{sec:idemp-line-arithm} is devoted to proving the following independence result.
\begin{theorem}
  \label{thm:5}
  Let \(n,m,k \in \mathbb{N}\) with \(0 < n < m\), then
  \[
    (\BLA + \axiomB{2} + \axiomB{3}) + \RINDParameterFree{\exists_{1}(\LLA)} \not \vdash \formulaE{k}{n}{m}.
  \]
\end{theorem}
By making use of the above independence result and the characterization of refutation by a clause set cycle in Proposition~\ref{pro:10}, we straightforwardly obtain an unrefutability result.
\begin{definition}
  \label{def:24}
  Let \(k, n, m \in\mathbb{N}\) with \(0 < n < m\), then we define the clause set \(\mathcal{E}_{k,n,m}(\eta)\) by \(\cnf(\BLA + \axiomB{2} + \axiomB{3} + \neg \formulaE{k}{n}{m}(\eta))\).

\end{definition}
\begin{corollary}
  \label{cor:6}
  Let \(k, n, m \in \mathbb{N}\) with \(0 < n < m\), then the clause set \(\mathcal{E}_{k,n,m}(\eta)\) is not refuted by an \(\LLA\) clause set cycle.
\end{corollary}
\begin{proof}
  Assume that \(\cnf(\BLA + \axiomB{2} + \axiomB{3} + \neg \formulaE{k}{n}{m}(\eta))\) is refuted by a clause set cycle.
  By Proposition \ref{pro:10} we have \([\BLA + \axiomB{2} + \axiomB{3}, \RINDParameterFreeEta{\exists_{1}(\LLA)}] \vdash \formulaE{k}{n}{m}(\eta)\).
  Since \((\BLA + \axiomB{2} + \axiomB{3} + \RINDParameterFree{\exists_{1}(\LLA)}) \vdash [\BLA + \axiomB{2} + \axiomB{3}, \RINDParameterFreeEta{\exists_{1}(\LLA)}]\), this contradicts Theorem~\ref{thm:5}.
\end{proof}
Let us now discuss this unprovability result.
The clause sets \(\mathcal{E}_{k, n, m}(\eta)\) with \(k,n,m \in \mathbb{N}\) and \(0 < n < m\) are refuted by open induction.
\begin{proposition}
  \label{pro:6}
  \(\IND{\Open(\LLA)} \cup \mathcal{E}_{k,n,m}(\eta)\) is unsatisfiable.
\end{proposition}
\begin{proof}
  Clearly, it suffices to show that \(B + \IND{\Open(\LLA)} \vdash \formulaE{k}{n}{m}(x)\).
  Work in \(B + \IND{\Open(\LLA)}\) and assume \(n\cdot x + \numeral{(m-n)k} = m\cdot x\).
  Then by \eqref{B2}, \eqref{B3}, and \eqref{B4} we obtain \(\numeral{(m-n)k} = (m-n)\cdot x\).
  Now we use \eqref{B1} to proceed by case analysis on \(x\).
  If \(x = \numeral{k'}\) with \(k' < k\), then we have \(\numeral{(m-n)(k - k')} = 0\).
  Since \(m - n > 0\) and \(k - k' > 0\) this contradicts Lemma~\ref{lem:3:1}.
  If \(x = \numeral{k}\), then we are done.
  If \(x = s^{k + 1}(p^{k + 1}(x))\), then \(0 = (m-n) + p^{k + 1}(x)\), which contradicts Lemma~\ref{lem:3:1}.
\end{proof}
Hence, Corollary~\ref{cor:6} together with Proposition~\ref{pro:6} give a positive answer to Conjecture~\ref{con:1}.
We conclude this section with some remarks on this result and possible improvements.

The formula \(E_{0,1,2}(x)\) is particularly interesting, because it can be proven by a comparatively straightforward induction.
\begin{lemma}
  \label{lem:13}
  \([\BLA, \RIND{\Open(\LLA)}] \vdash E_{0,1,2}\).
\end{lemma}
\begin{proof}
  Clearly it suffices to show that the formula \(\varphi(x,y) \coloneqq  x + 0 = y + x \rightarrow y = 0\) is \(\BLA\)-inductive in \(x\).
  It is obvious that \(\BLA \vdash \varphi(0,y)\).
  Now work in \(\BLA\) and assume \(\varphi(x,y)\) and \(s(x) + 0 = y + s(x)\).
  By \(\eqref{Q_3}\) and \(\eqref{Q_4}\) we obtain \(s(x + 0) = s(x) = s(x) + 0 = y + s(x) = s(y + x)\).
  By Lemma \ref{lem:3} we obtain \(x + 0 = y + x\), hence by the assumption we obtain \(y = 0\).
\end{proof}
This demonstrates that clause set cycles are a very weak induction mechanism in the sense that they are unable to deal even with simple generalizations and therefore fail to refute relatively simple clause sets.
The unprovability results in Corollaries~\ref{cor:3} and \ref{cor:6} were constructed so that only one Skolem constant \(\eta\) appears in the language of the considered clause sets.
Consider now the clause set \(\mathcal{C}\) given by
\[
  \cnf(\BLA) \cup \{ \{ \eta + 0 = \nu + \eta \} \} \cup \{ \{ \nu \neq 0 \} \},
\]
where \(\nu\) is a Skolem constant distinct from \(\eta\).
It is straightforward to check that \(\mathcal{C}(\eta)\) is an \(\LLA \cup \{ \nu \}\) clause set cycle.
Hence, if clause set cycles are detected on the languages obtained by Skolemization of the given property and its background theory, then clause set cycles allow us to prove the property \(x + 0 = y + x \rightarrow y = 0\) from \(B\) but fail to prove the weaker property \(x + 0 = x + x \rightarrow x = 0\) from \(B\).
Thus, clause set cycles are sensitive to the syntactic material present in a given set clauses.
In particular, Skolem constants other than \(\eta\) may act similar to induction parameters.

The independence result of Theorem~\ref{thm:5} also shows that the unrefutability result of Corollary~\ref{cor:6} does neither rely on the \(\eta\)-restriction nor on the absence of nesting in clause set cycles.
Moreover, in the light of Lemma \ref{lem:13} we conjecture that the unrefutability of Corollary~\ref{cor:6} is entirely due to the absence of induction parameters from induction captured by clause set cycles.
\begin{conjecture}
  \label{con:2}
  Let \(k,n, m \in \mathbb{N}\) with \(0 < n < m\), then
  \[
    B + \axiomB{2} + \axiomB{3} + \INDParameterFree{\exists_{1}(\LLA)} \not \vdash \formulaE{k}{n}{m}.
  \]
\end{conjecture}
Furthermore, we believe that an independence similar to the one in Conjecture~\ref{con:2} also holds for the atomic formula \(x + (x + x) = (x + x) + x\), which is a well-known challenging formula for inductive theorem provers \cite{baker1992,beeson2006,hajdu2020}.
\begin{conjecture}
  \label{con:ind_ex_1_param_free_not_proves_3x_eq_x3}
  \(\BLA + \INDParameterFree{\exists_{1}(\LLA)} \not \vdash x + (x + x) = (x + x) + x\).
\end{conjecture}
\subsubsection{Nesting of the induction rule}
In this section we briefly consider the role of the depth of the nesting of applications of the induction rule.
The idea underlying the results developed in this section was brought to our attention by one of the anonymous reviewers. 
We will show that a formalism that extends clause set cycles to achieve a fixed finite depth of the nesting of the corresponding induction rule will have an unprovable clause set, that becomes provable when the nesting depth is increased by one.
Moreover, the result remains valid in extensions of clause set cycles that allow for induction parameters.
However, the unprovability results in this section are more abstract than in the previous sections in the sense that we work over a much stronger background theory.
We expect that providing more elementary unprovability results is not difficult but is left as future work.

In the remainder of the section we will show the following result.
\begin{theorem}
  \label{thm:separation_nesting_of_induction_rule}
  Let \(k \in \NaturalNumbers\), then there is a language \(L\) and an \(L \cup \{ \eta \}\) clause set \(\mathcal{C}(\eta)\) such that \(\mathcal{C}\) is consistent with \([\varnothing, \RINDParameterFree{\exists_{1}(L)}]_{k}\) but inconsistent with \([\varnothing, \RINDParameterFree{\exists_{1}(L)}]_{k + 1}\).
\end{theorem}
\newcommand{\LanguagePeanoArithmetic}{L_{\mathrm{PA}}}
The language of Peano arithmetic \(\LanguagePeanoArithmetic\) consists of the function symbols \(0/0\), \(s/1\), the infix function symbols \(+/2\), \(*/2\), and the infix predicate symbol \(\leq/2\).
Let \(x\) and \(y\) be distinct variables, then we write \(\Exists{x \leq y}{\varphi}\) as an abbreviation for \(\Exists{x}{\left(x \leq y \wedge \varphi\right)}\) and similarly we write \(\Forall{x \leq y}{\varphi}\) as an abbreviation for the formula \(\Forall{x}{\left(x \leq y \rightarrow \varphi\right)}\).
A \(\LanguagePeanoArithmetic\) formula is said to be bounded if all the quantifiers occurring in it are bounded as above.
The \(\Sigma_{0}\), \(\Pi_{0}\) and \(\Delta_{0}\) formulas are the bounded formulas.
The \(\Sigma_{n + 1}\) (\(\Pi_{n + 1}\)) formulas are the formulas of the form \(\Exists{\vec{x}}{\varphi}\) (\(\Forall{\vec{x}}{\varphi}\)) where \(\vec{x}\) is a possibly empty finite sequence of variables and \(\varphi\) is a \(\Pi_{n}\) (\(\Sigma_{n}\)) formula.

We will prove the theorem above by providing a sequence of theories \(T_{0}, T_{1}, \dots\) with \(L(T_{i}) \supseteq \LanguagePeanoArithmetic\), \(T_{i + 1} = [T_{i}, \RINDParameterFree{\exists_{1}(L(T_{0}))}]\) such that the provably total recursive functions of \(T_{i}\) are exactly those of the level \(3 + i\) of the Grzegorczyk hierarchy, for \(i \in \NaturalNumbers\), and over \(T_{0}\) the \(\Sigma_{1}\) formulas are exactly the \(\exists_{1}(\LanguageOf{T_{0}})\) formulas.
Since, the Grzegorczyk hierarchy is a strict hierarchy (see for example \cite{rose1984}), we obtain for each level \(i \in \NaturalNumbers\) a quantifier-free \(\LanguageOf{T_{0}}\) formula \(\varphi(x,y)\), such that \(\Exists{y}{\varphi(x,y)}\) is provable in \(T_{i + 1}\) but not in \(T_{i}\).

For a definition of the Grzegorczyk hiearchy we refer the reader to \cite{rose1984}.
\newcommand{\GrzegorczykClass}[1]{\mathcal{E}_{#1}}
\begin{definition}
  \label{def:grzegorczyk_levels}
  Let \(n \in \NaturalNumbers\), then we denote by \(\GrzegorczykClass{n}\) the \(n\)-th level of the Grzegorczyk hiearchy.
\end{definition}
\newcommand{\RobinsonQ}{Q}
\begin{definition}
  \label{def:robinson_arithmetic}
  The theory \(Q\) is axiomatized by the universal closure of the following axioms
  \begin{align*}
    s(x) & \neq 0, \tag{Q1} \label{ax:Q:1}
    \\
    s(x) = s(y) & \rightarrow x = y, \tag{Q2} \label{ax:Q:2}
    \\
    x \neq 0 & \rightarrow \Exists{y}{\left( x = s(y)\right)}, \tag{Q3} \label{ax:Q:3}
    \\
    x + 0 & = x, \tag{Q4} \label{ax:Q:4}
    \\
    x + s(y) & = s(x + y), \tag{Q5} \label{ax:Q:5}
    \\
    x * 0 & = 0, \tag{Q6} \label{ax:Q:6}
    \\
    x * s(y) & = (x * y) + x, \tag{Q7} \label{ax:Q:7}
    \\
    x \leq y & \leftrightarrow \Exists{z}{\left(z + x = y\right)}. \tag{Q8} \label{ax:Q:8}
  \end{align*}
\end{definition}
\begin{definition}
  \label{def:theory_isigma_n}
  Let \(n \in \mathbb{N}\), then the theory \(Q + \IND{\Sigma_{n}}\) is called \(I\Sigma_{n}\).
  The theory \(I\Sigma_{0}\) is also called \(I\Delta_{0}\).
\end{definition}
There is a \(\Delta_{0}\) definition of the exponential function such that the theory \(I\Delta_{0}\) proves the inductive properties of the definition of the exponential function, but \(I\Delta_{0}\) does not prove the totality of such a definition.
\newcommand{\ExpDF}[3]{\mathrm{Exp}(#1,#2,#3)}
\begin{lemma}
  \label{lem:exponentiation_relation_in_idelta_0}
  There is a \(\Delta_{0}\) formula \(\ExpDF{x}{y}{z}\) such that \(I\Delta_{0}\) proves
  \begin{align}
    \ExpDF{x}{0}{z} & \leftrightarrow z = \numeral{1}, \tag{E1} \label{exp:1}
    \\
    \ExpDF{x}{s(y)}{z} & \leftrightarrow \Exists{v}{\left(\ExpDF{x}{y}{v} \wedge z = v * x\right)}. \tag{E2} \label{exp:2}
  \end{align}
  In particular \(I\Delta_{0}\) proves \(\ExpDF{x}{y}{z_{1}} \wedge \ExpDF{x}{y}{z_{2}} \rightarrow z_{1} = z_{2}\).
\end{lemma}
\begin{proof}
  See \cite[Section~V.3]{hajek1993}
\end{proof}
In the following we will mainly work a theory that extends \(I\Delta_{0}\) by a statement asserting the totality of the exponential function.
\begin{definition}
  \label{def:idelta_0_plus_exponentiation}
  By \(I\Delta_{0} + \mathrm{EXP}\) we denote the theory that extends \(I\Delta_{0}\) by the axiom \(\Forall{x}{\Forall{y}{\Exists{z}{\ExpDF{x}{y}{z}}}}\).
\end{definition}
The theory \(I\Delta_{0} + \mathrm{EXP}\) is also called elementary arithmetic and has various equivalent formulations, see \cite[Section~1.1]{beklemishev2005}.
In the following we will develop a particular formulation with a \(\forall_{1}\) axiomatization and in which the \(\exists_{1}\) formulas of the extended language are exactly the \(\Sigma_{1}\) formulas.
\begin{lemma}
  \label{lem:idelta_0_has_a_pi_1_axiomatization}
  \(I\Delta_{0}\) has a \(\Pi_{1}\) axiomatization.
\end{lemma}
\begin{proof}
  Drop axiom \ref{ax:Q:3}, replace axiom \ref{ax:Q:8} by the universal closure of the formulas \(x \leq y \rightarrow \Exists{z \leq y}{z + x = y}\) and \(z + x = y \rightarrow x \leq y\), and replace the induction axioms \(I_{x}\varphi\), where \(\varphi(x,\vec{z})\) is \(\Delta_{0}\) by
  \[
    \left(\varphi(0,\vec{z}) \wedge \Forall{y < x}{\left(\varphi(y,\vec{z}) \rightarrow \varphi(s(y),\vec{z}) \right)} \right) \rightarrow \varphi(x,\vec{z}).
  \]
  It is routine to check that the resulting theory is equivalent to \(I\Delta_{0}\).
\end{proof}
Now we will show that \(I\Delta_{0}\) has \(\Delta_{0}\) definitions of Skolem functions of all \(\Delta_{0}\) formulas.
Later on we will introduce the corresponding Skolem functions in order to get rid of bounded quantifiers.
\begin{definition}[Least number principle]
  \label{def:4}
  Let \(\varphi(x,\vec{z})\) be a formula, then the least number principle for \(\varphi\) is given by
  \[
    \Exists{x}{\varphi(x,\vec{z})} \rightarrow \Exists{x}{\left(\varphi(x,\vec{z}) \wedge \Forall{y < x}{\neg \varphi(y,\vec{z})}\right)}.
  \]
\end{definition}
\begin{lemma}
  \label{lem:idelta_0_proves_least_number_principle_for_delta_0_formulas}
  \(I\Delta_{0}\) proves the least number principle for \(\Delta_{0}\) formulas.
\end{lemma}
\begin{proof}
  See \cite[Theorem~1.22]{hajek1993}.
\end{proof}
\begin{definition}
  \label{def:delta_0_skolem_function_delta_0_definition}
  Let \(\varphi(\vec{x},y,z)\) be a \(\Delta_{0}\) formula, then the formula \(D_{\Exists{z \leq y}{\varphi}}(\vec{x},y,z)\) is given by
  \[
    \left(z \leq y \wedge \varphi(x,y,z) \wedge \Forall{z' < z}{\neg \varphi(\vec{x},y,z')}\right)
    \vee \left(\Forall{z \leq y}{\neg \varphi(\vec{x},y,z)} \wedge z = 0 \right).
  \]
\end{definition}
\begin{lemma}
  \label{lem:properties_of_definition_of_skolem_function}
  Let \(\varphi(\vec{x},y,z)\) be a \(\Delta_{0}\) formula, then \(I\Delta_{0}\) proves
  \begin{enumerate}[label=(\roman*),ref =\ref{lem:properties_of_definition_of_skolem_function}{.}{\it (\roman*)}]
  \item \(\Exists{z \leq y}{\varphi(\vec{x},y,z)} \rightarrow \Exists{z \leq y}{(D_{\Exists{z \leq y}{\varphi}}(\vec{x},y,z) \wedge \varphi(\vec{x},y,z))}\) \label{lem:properties_of_definition_of_skolem_function:1}
  \item \(\Forall{z \leq y}{\neg \varphi(\vec{x},y,z)} \rightarrow D_{\Exists{z \leq y}{\varphi}}(\vec{x},y,0)\) \label{lem:properties_of_definition_of_skolem_function:2}
  \item \(\ExistsExOne{z}{D_{\Exists{z \leq y}{\varphi}}(\vec{x},y,z)}\), \label{lem:properties_of_definition_of_skolem_function:3}
  \item \(D_{\Exists{z \leq y}{\varphi}}(\vec{x},y,z) \rightarrow z \leq y\). \label{lem:properties_of_definition_of_skolem_function:4}
  \end{enumerate}
\end{lemma}
\begin{proof}
  The formula \ref{lem:properties_of_definition_of_skolem_function:1} follows easily from Lemma~\ref{lem:idelta_0_proves_least_number_principle_for_delta_0_formulas}.
  The formulas \ref{lem:properties_of_definition_of_skolem_function:2}--\ref{lem:properties_of_definition_of_skolem_function:4} are straightforward.
\end{proof}
We will now define the way in which we Skolemize \(\Delta_{0}\) formulas.
\begin{definition}
  \label{def:delta_zero_skolem_symbol}
  Let \(\varphi(\vec{x},y,z)\) be a \(\Delta_{0}\) formula, then \(F_{\Exists{z \leq y}{\varphi}}\) is a function symbol of arity \(|\vec{x}| + 1\).
\end{definition}
\begin{definition}
  \label{def:skolemization_functions}
  The formula translations \((\cdot)^{\exists}\) and \((\cdot)^{\forall}\) are defined mutually recursively by
  \begin{gather*}
    (\theta)^{Q} = \theta, \text{if \(\theta\) is quantifier-free},
    \\
    (\varphi_{1} \wedge \varphi_{2})^{Q} = \varphi_{1}^{Q} \wedge \varphi_{2}^{Q},
    \\
    (\varphi_{1} \vee \varphi_{2})^{Q} = \varphi_{1}^{Q} \vee \varphi_{2}^{Q},
    \\
    (\neg \varphi)^{Q} = \neg \varphi^{\overline{Q}},
    \\
    (\Quantifier{\overline{Q}}{y \leq x}{\varphi})^{Q} = \Quantifier{\overline{Q}}{y \leq x}{\varphi^{Q}}
    \\
    (\Exists{y \leq x}{\varphi(x,y,\vec{z})})^{\exists} = \left(y \leq x \wedge \varphi^{\exists}\right)[y/F_{\Exists{y \leq x}{\varphi}}(x,\vec{z})], \tag{\(\ast_{1}\)} \label{sk:exists}
    \\
    (\Forall{y \leq x}{\varphi(x,y,\vec{z})})^{\forall} = \left(y \leq x \rightarrow \varphi^{\forall}\right)[y/F_{\Exists{y \leq x}{\neg \varphi}}(x,\vec{z})], \tag{\(\ast_{2}\)} \label{sk:forall}
  \end{gather*}
  where \(Q \in \{ \forall, \exists \}\), \(\overline{\forall} = \exists\), \(\overline{\exists} = \forall\), and in Equations~\eqref{sk:exists} and \eqref{sk:forall} the variables \(\vec{z}\) all appear freely in the formula \(\varphi\).
\end{definition}
We can now obtain a suitable formulation of \(I\Delta_{0} + \mathrm{EXP}\).
\begin{lemma}
  \label{lem:idelta_0_plus_exponentation_has_a_suitable_formulation}
  There exists a \(\forall_{1}\) axiomatized conservative extension \(T\) of \(I\Delta_{0} + \mathrm{EXP}\) such that every \(\exists_{1}(L(T))\) formula is equivalent over \(T\) to a \(\Sigma_{1}\) formula and every \(\Sigma_{1}\) formula is equivalent over \(T\) to an \(\exists_{1}(L(T))\) formula.
\end{lemma}
\begin{proof}
  We consider a \(\Pi_{1}\) formulation \(U\) of \(I\Delta_{0}\).
  For each axiom \(\Forall{\vec{x}}{\varphi}\) of \(U\) where \(\varphi\) is \(\Delta_{0}\), \(T\) contains the axiom \(\Forall{\vec{x}}{\varphi^{\exists}}\).
  Furthermore, \(T\) contains the axiom \(\mathrm{Exp}^{\exists}[z/e(x,y)]\).
  Finally, for each \(\Delta_{0}\) formula \(\varphi(\vec{x},y,z)\), \(T\) contains the axiom \((D_{\Exists{z \leq y}{\varphi}})^{\exists}[z/F_{\Exists{z \leq y}{\varphi}}(\vec{x},y)]\).
  Now obtain a \(\forall_{1}\) axiomatization by moving the remaining quantifiers outwards.
  By a model-theoretic argument it is straightforward to see that the resulting theory is conservative over \(I\Delta_{0} + \mathrm{EXP}\).

  It is straightforward to check that every \(\Delta_{0}\) formula \(\varphi\) is equivalent in \(T\) to a quantifier-free \(L(T)\) formula.
  Let \(\psi\) be a \(\Sigma_{1}\) formula, then \(\psi = \Exists{\vec{x}}{\varphi}\) where \(\varphi\) is \(\Delta_{0}\).
  Hence, \(\psi\) is equivalent over \(T\) to the formula \(\Exists{\vec{x}}{\varphi'}\) where \(\varphi'\) is a quantifier-free formula that is equivalent over \(T\) to \(\varphi\).
  Now let \(\psi\) be an \(\exists_{1}(L(T))\) formula, then by \cite[pp.\ 51--52]{hodges1997} there exists an equivalent unnested \(\exists_{1}(L(T))\) formula of the form \(\Exists{\vec{x}}{\varphi}\) where \(\varphi\) is quantifier-free.
  Now we simply replace atoms of the form \(f(\vec{u}) = y\) where \(f\) is either a Skolem symbol of a \(\Delta_{0}\) formula or \(e\) by the corresponding defining \(\Delta_{0}\) formula.
  Hence, the resulting formula is a \(\Sigma_{1}\) formula.
\end{proof}
In the following we fix one such extension of \(I\Delta_{0} + \mathrm{EXP}\) and call it \(\mathrm{EA}\).
\begin{definition}
  \label{def:grzegorczyk_arithmetics}
  Let \(k \in \NaturalNumbers\), then \(\mathrm{EA}_{k}\) denotes the theory \([\mathrm{EA}, \RIND{\Pi_{2}}]_{k}\).
\end{definition}
\begin{theorem}[\cite{sieg1991}]
  \label{thm:1}
  The provably total recursive functions of the theory \(\mathrm{EA}_{k}\) are precisely those of the class \(\mathcal{E}_{3 + k}\) of the Grzegorczyk hierarchy.
\end{theorem}
\begin{proof}
  See also the proof Corollary~7.5 of \cite{beklemishev1997}.
\end{proof}
We can reformulate the theories \(\mathrm{EA}_{k}\) as follows.
\begin{lemma}
  \label{thm:reformulation_of_grzegorczyk_arithmetics}
  Let \(k \in \NaturalNumbers\), then \(\mathrm{EA}_{k} \equiv [\mathrm{EA}, \RINDParameterFree{\exists_{1}(L(\mathrm{EA}))}]_{k}.\)
\end{lemma}
\begin{proof}
  We proceed by induction on \(k\) and show
    \[
    \mathrm{EA}_{k} \equiv [\mathrm{EA},\RINDParameterFree{\exists_{1}(L(\mathrm{EA}))}]_{k}.
  \]
  If \(k = 0\), then the claim follows trivially.
  Now assume the claim for \(k\), then \(\mathrm{EA}_{k}\) is \(\Pi_{2}\) axiomatized, hence by \cite[Corollary~7.4]{beklemishev1997}
  \[
    \mathrm{EA}_{k + 1} \equiv [\mathrm{EA}_{k}, \RIND{\Pi_{2}}] \equiv [\mathrm{EA}_{k}, \RIND{\Sigma_{1}}].
  \]
  Furthermore, by \cite[Lemma~4.6]{beklemishev2005} we have
  \[
    [\mathrm{EA}_{k}, \RIND{\Sigma_{1}}] \equiv [\mathrm{EA}_{k}, \RINDParameterFree{\Sigma_{1}}].
  \]
  Since over \(\mathrm{EA}\) the \(\Sigma_{1}\) formulas are exactly the \(\exists_{1}(L(\mathrm{EA}))\) formulas, we obtain
  \[
    [\mathrm{EA}_{k}, \RINDParameterFree{\Sigma_{1}}] \equiv [\mathrm{EA}_{k}, \RINDParameterFree{\exists_{1}(L(\mathrm{EA}))}].
  \]
  By the induction hypothesis we readily obtain \[
    [\mathrm{EA},\RIND{\Pi_{2}}]_{k + 1} \equiv [\mathrm{EA}, \RINDParameterFree{\exists_{1}(L(\mathrm{EA}))}]_{k + 1}. \qedhere
  \]
\end{proof}
Since \(\mathcal{E}_{k} \subsetneq \mathcal{E}_{k + 1}\) for all \(k \in \NaturalNumbers\), we can now provide a proof of Theorem~\ref{thm:separation_nesting_of_induction_rule}.
\begin{proof}[Proof of Theorem~\ref{thm:separation_nesting_of_induction_rule}]
  Let \(k \in \NaturalNumbers\), then there exists a function \(f: \NaturalNumbers \to \NaturalNumbers\) such that \(f \in \mathcal{E}_{k + 4} \setminus \mathcal{E}_{k + 3}\).
  Hence, there exists a \(\Sigma_{1}\) formula \(\varphi(x,y)\) such that \(f(n) = m\) if and only if \(\NaturalNumbers \models \varphi(\numeral{n},\numeral{m})\) and
  \begin{gather*}
    [\mathrm{EA},\RINDParameterFree{\exists_{1}(L(\mathrm{EA}))}]_{k + 1} \vdash \Exists{y}{\varphi(x,y)},
    \\
    [\mathrm{EA},\RINDParameterFree{\exists_{1}(L(\mathrm{EA}))}]_{k} \not \vdash \Exists{y}{\varphi(x,y)}.
  \end{gather*}
  Thus, by the construction of \(\mathrm{EA}\), there exists a quantifier-free \(L(\mathrm{EA})\) formula \(\varphi'(x,y,\vec{z})\) such that \(\mathrm{EA} \vdash \varphi \leftrightarrow \Exists{\vec{z}}{\varphi'} \).
  Since \(\mathrm{EA}\) is \(\forall_{1}\) axiomatized we furthermore have \(\mathrm{EA} + [\varnothing,\RINDParameterFree{\exists_{1}(L(\mathrm{EA}))}] \equiv [\mathrm{EA},\RINDParameterFree{\exists_{1}(L(\mathrm{EA}))}]\).
  Hence, \(\mathcal{C} \coloneqq \cnf(\mathrm{EA} + \Forall{y}{\Forall{\vec{z}}{\neg \varphi'(\eta,y,\vec{z})}})\) is a suitable clause set.
\end{proof}
This result tells us that a mechanism that extends refutation by a clause set cycle so as to allow at most \(k\)-fold nested \(\exists_{1}\) parameter-free induction rule is strictly weaker than a mechanism that allows \((k + 1)\)-fold nested applications of the \(\exists_{1}\) parameter-free induction rule.
This naturally gives rise to the question whether we can separate a system that provides arbitrary nestings of the parameter-free \(\exists_{1}\) induction rule from a system that provides the parameter-free \(\exists_{1}\) induction schema.
The following lemma shows that we need a different approach to resolve this question.
\begin{lemma}[\cite{parsons1972}]
  \label{lem:i_sigma_1_pi_2_conservative_over_ea_sigma_1_ir}
  \(I\Sigma_{1}\) is \(\Pi_{2}\) conservative over \(\mathrm{EA} + \RIND{\Sigma_{1}}\).
\end{lemma}
Hence the theory \(\mathrm{EA} + \IND{\exists_{1}(L(\mathrm{EA}))}\) is also \(\Pi_{2}\) conservative over \(\mathrm{EA} + \RINDParameterFree{L(\mathrm{EA})}\).
Thus the technique used above does not provide us with a clause set that separates both systems.

Nevertheless, we conjecture that the parameter-free \(\exists_{1}\) induction schema is in general stronger than the parameter-free \(\exists_{1}\) induction rule.
\begin{conjecture}
  \label{con:axiom_stronger_than_rule}
  There exists a language \(L\) and an \(L \cup \{ \eta \}\) clause set \(\mathcal{D}(\eta)\) such that \(\INDParameterFree{\exists_{1}(L)} + \mathcal{D}(\eta)\) is inconsistent, but \((\varnothing + \RINDParameterFree{\exists_{1}(L)}) + \mathcal{D}(\eta)\) is consistent.
\end{conjecture}
The results in this section are less elementary than the results of Sections~\ref{sec:unrefutability_instance_restriction} and \ref{sec:unrefutability_rule_and_parameters} in the sense that we work over the comparatively strong \(\mathrm{EA}\) and the separation involves clause sets that express totality assertions.
However, totality assertions are an important class of problems for AITP systems.
In this sense the connection with the Grzegorczyk hierarchy is remarkable.
\section{Idempotents in linear arithmetic}
\label{sec:idemp-line-arithm}
In the previous section we have introduced clause set cycles and we have given a characterization of refutation by a clause set cycle in terms of a logical theory.
Moreover, we have shown two unrefutability results for clause set cycles.
We have shown the second unrefutability result by anticipating the independence result of Theorem~\ref{thm:5} for which a proof will be provided in this section.
In Section~\ref{sec:syntatic_simplifications} we introduce some preliminary notions and we carry out some syntactic simplifications on \(\exists_{1}\) formulas.
In Section~\ref{sec:linear-systems} we consider some properties of \(\exists_{1}\) formulas in the structures \(\mathbb{N}\) and \(\mathbb{Z}\).
Finally, in Section~\ref{sec:model-theor-constr} we carry out the model theoretic construction.

We work in the setting of linear arithmetic, hence, unless stated otherwise, whenever we speak of a formula (sentence) we mean an \(\LLA\) formula (sentence).
\subsection{Preliminaries}
\label{sec:syntatic_simplifications}
In this section we mainly carry out some syntactic transformations that allow us to eliminate the function symbols \(p\) and \(0\) from \(\exists_{1}\) formulas.
The absence of these symbols allows us to carry out certain embeddings of structures in Sections~\ref{sec:linear-systems} and \ref{sec:model-theor-constr}.
\begin{definition}
  \label{def:7}
  The theory \(\theoryV\) is axiomatized by the universal closure of the formulas
    \begin{equation}
    \label{eq:9}
    \numeral{k} + x = x + \numeral{k}, \tag{\(\axiomV{k}\)}
  \end{equation}
  where \(k \in \mathbb{N}\).
\end{definition}
\begin{lemma}
  \label{lem:1}
  \([B, \RINDParameterFree{\Open(\{ 0, s, +\})}] \vdash \theoryV\).
\end{lemma}
\begin{proof}
  The formula \(\numeral{k} + x = s^{k}(x)\) is \(\BLA\)-inductive and furthermore \(\BLA \vdash s^{k}(x) = x + \numeral{k}\).
\end{proof}
We will carry out these transformations in the very weak theory \(\BLA + \axiomB{1} + \theoryV\).
In a first step we will show that we can eliminate the symbol \(p\) from \(\exists_{1}\) formulas without increasing the quantifier complexity of \(\exists_{1}\) formulas.
After that, we show that we can moreover eliminate to a certain extent the symbol \(0\) from \(\exists_{1}\) formulas, again without increasing the quantifier complexity.

In order to eliminate the symbol \(p\) from \(\exists_{1}\) formulas we proceed by replacing all the occurrences of the symbol \(p\) by the following definition of the predecessor function.
\begin{definition}
  \label{def:5}
  We define the formula \(\mathrm{D}(x,y)\) by
  \begin{gather*}
    (x = 0 \wedge y = 0) \vee s(y) = x.
  \end{gather*}
\end{definition}
\begin{lemma}
  \label{lem:11}
  \(\BLA + \mathrm{B1} \vdash p(x) = y \leftrightarrow \mathrm{D}(x, y)\).
\end{lemma}
\begin{proof}
  We work in \(\BLA + \mathrm{B1}\).
  Assume \(p(x) = y\).
  If \(x = 0\), then we have \(y = p(x) = p(0) = 0\), hence \(\mathrm{D}(x, y)\).  
  Otherwise, \(x = s(p(x))\) and therefore \(x = s(p(x)) = s(y)\).
  Now assume \(\mathrm{D}(x, y)\).
  If \(x = 0 \wedge y = 0\), then we have \(p(x) = p(0) = 0 = y\).
  If \(s(y) = x\), then \(y = p(s(y)) = p(x)\).
\end{proof}
We can now factor the symbol \(p\) into the axiom \(\axiomB{1}\) by replacing all the occurrences of \(p\) by the definition of the predecessor function.
\begin{lemma}
  \label{lem:8}
  Let \(\varphi(\vec{x})\) be an \(\exists_{1}\) \(\LLA\) formula, then there exists a \(p\)-free \(\exists_{1}\) formula \(\varphi'(\vec{x})\) such that
  \[
    \BLA + \mathrm{B1} \vdash \varphi \leftrightarrow \varphi'.
  \]  
\end{lemma}
\begin{proof}
  Let \(\varphi\)  be an \(\exists_{1}(\LLA)\) formula, then there exists an unnested \(\exists_{1}(\LLA)\) formula \(\psi\) such that \(\vdash \varphi \leftrightarrow \psi\), see for example \cite[pp.\ 51--52]{hodges1997}.
  In particular, the symbol \(p\) occurs in \(\psi\) only in atoms of the form \(p(x) = y\).
  Hence, we obtain the desired formula by replacing in \(\psi\) the atomic formulas of the form \(p(x) = y\) by \(\mathrm{D}(x,y)\).
\end{proof}
In the following we will eliminate the symbol \(0\) to a certain extent from \(\exists_{1}\) formulas in one variable.
In order to simplify the arguments we will introduce some additional assumptions.
Since we work in the context of the theory \(\BLA\) we can by Lemma~\ref{lem:33} assume without loss of generality that ground terms are numerals.
Moreover, since equality is symmetric we will assume without loss of generality that atoms are oriented in such a way that whenever the atom contains a variable, then the left hand side of the atom contains a variable.

Let us start by introducing the notion of components, a class of \(\exists_{1}\) formulas that is particularly suitable to carry out the elimination of the symbol \(0\).
Moreover, components will also be of use for the arguments in Section~\ref{sec:linear-systems}.
\begin{definition}[Components]
  \label{def:13}
  A component \(\chi(\vec{x})\) is a formula of the form \(\exists \vec{y} C_{\chi}(\vec{x}, \vec{y})\), where \(C_{\chi}\) is a conjunction of literals.
\end{definition}
\begin{lemma}
  \label{lem:exists_one_formula_to_p_free_components}
  Let \(\varphi(x)\) be an \(\exists_{1}\) formula, then there exist \(p\)-free components \(\chi_{1}, \dots, \chi_{n}\) such that \(\BLA + \axiomB{1} \vdash \varphi \leftrightarrow \bigvee_{i = 1}^{n}\chi_{i}\).
\end{lemma}
\begin{proof}
  Apply Lemma~\ref{lem:8} to obtain a \(p\)-free \(\exists_{1}\) formula \(\varphi'\) such that \(\BLA + \axiomB{1} \vdash \varphi \leftrightarrow \varphi'\).
  Now obtain the desired components by replacing formulas of the form \(\varphi \rightarrow \psi\) and \(\varphi \leftrightarrow \psi\) respectively by \(\neg \varphi \vee \psi\) and \((\neg \varphi \vee \psi) \wedge (\neg \psi \vee \varphi) \), moving negations inward, eliminating double negations, distributing conjunctions over disjunctions, and moving existential quantifiers inwards over disjunctions.
\end{proof}
We will distinguish between three types of literals: Those where both sides contain variables, those where only one side of the equation contains a variable and those where none of the sides contain a variable.
\begin{definition}
  \label{def:27}
  Let \(l\) be a literal of the form \(u \bowtie v\) with \(\bowtie \ \in \{ =, \neq \}\), then \(l\) is: \(\UU\) if both \(u\) and \(v\) contain a variable, \(\UD\) if \(u\) contains a variable and \(v\) is  ground, and \(\DD\) if both \(u\) and \(v\) are ground.
  We will combine this notation with superscript \(+\) to indicate that the literal is positive and a superscript \(-\) to indicate that the literal is negative.
  We say that a \(\UD\) literal is simple if it is of the form \(z = \numeral{k}\) where \(z\) is a variable and \(k \in \mathbb{N}\) and complex otherwise.
\end{definition}
\begin{lemma}
  \label{lem:28}
  Let \(t\) be a term with \(\Var(t) \neq \varnothing\), then there exists a \(0\)-free term \(t'\) such that \(\BLA + \mathcal{V} \vdash t = t'\).
\end{lemma}
\begin{proof}
  We proceed by induction on the structure of the term \(t\).
  If \(t\) is a variable, then we are done by letting \(t' = t\).
  If \(t\) is of the form \(s(u)\), then \(\Var(u) \neq \varnothing\).
  Hence, we can apply the induction hypothesis to \(u\) in order to obtain a \(0\)-free term \(u'\) such that \(\BLA + \theoryV \vdash u = u'\). Thus, \(\BLA + \theoryV \vdash t = s(u')\) and we let \(t' = s(u')\).
  If \(t\) is of the form \(p(u)\), then we proceed analogously.
  If \(t\) is of the form \(u + v\), then we need to consider several cases.
  If \(\Var(u) = \varnothing\), then \(\Var(v) \neq \varnothing\) and we have \(\BLA \vdash u = \numeral{k}\) for some \(k \in \mathbb{N}\) and therefore \(\BLA + \theoryV \vdash t = \numeral{k} + v = v + \numeral{k} = v' + \numeral{k} = s^{k}(v')\).
  If \(\Var(v) = \varnothing\), then \(\Var(u) \neq \varnothing\).
  Hence, we can apply the induction hypothesis to \(u\) in order to obtain a \(0\)-free term \(u'\) such that \(\BLA + \theoryV \vdash u = u'\).
  Since \(\Var(v) = \varnothing\), there exists \(k \in \mathbb{N}\) such that \(\BLA \vdash v = \numeral{k}\).
  By multiple applications of \eqref{Q_4} followed by an application of \eqref{Q_4} we obtain \(\BLA + \theoryV \vdash t = u + \numeral{k} = s^{k}(u) + 0 = s^{k}(u)\).
  Hence, \(t' \coloneqq s^{k}(u)\) is the desired \(0\)-free term.
  If \(u\) and \(v\) contain variables, then by the induction hypothesis we obtain \(0\)-free terms \(u'\) and \(v'\) such that \(\BLA + \theoryV \vdash u = u'\) and \(\BLA + \theoryV \vdash v = v'\).
  Hence, \(t = u' + v'\) is the desired \(0\)-free term.
\end{proof}
By Lemma~\ref{lem:34} and Lemma~\ref{lem:28}, it is straightforward to eliminate the symbol \(0\) from \(\UU\) and \(\DD\) literals.
However, eliminating the symbol \(0\) from \(\UD\) literals needs some more work.
Let us start by observing that complex \(\UD\) atoms can be split into several simple ones.
\begin{lemma}
  \label{lem:split_complex_up_down_atom}
  Let \(u(\vec{z})\) be a \(p\)-free term with \(\vec{z} = (z_{1}, \dots, z_{l})\) and \(k \in \mathbb{N}\), then
  \[
    \BLA + \mathrm{B1} \vdash u(\vec{z}) = \numeral{k} \leftrightarrow
    \bigvee_{\substack{
      0 \leq m_{1}, \dots, m_{l} \leq k \\
      \mathbb{N} \models u(m_{1}, \dots, m_{l}) = k
    }}\bigwedge_{j = 1}^{l} z_{j} = \numeral{m_{j}}.
  \]
\end{lemma}
\begin{proof}
  Work in \(\BLA + \mathrm{B1}\).
  The ``if'' direction is obvious.
  For the ``only if'' direction assume \(u(\vec{z}) = \numeral{k}\) and proceed by \(k\)-fold case analysis on the variables \(\vec{z}\).
  If \(z_{i} = \numeral{m_{i}}\) with \(0 \leq m_{i} \leq k\) for \(i = 1, \dots, l\), then we have two cases.
  If \(u(\numeral{m_{1}}, \dots, \numeral{m_{l}}) \neq \numeral{k}\), then the claim follows trivially.
  Otherwise if \(u(\numeral{m_{1}}, \dots, \numeral{m_{l}}) = \numeral{k}\), then we are done as well since \(z_{1} = \numeral{m_{1}} \wedge \dots \wedge z_{l} = \numeral{m_{l}}\) is a conjunct of the right side.
  Otherwise, there exists an \(i \in \{ 1, \dots, l\}\) and \(z_{i}'\) such that \(z_{i} = s^{k+1}z_{i}'\).
  Then let \(j\) be the index of the variable \(z_{j}\) with the rightmost occurrence such that \(z_{j} = s^{k+1}z_{j}'\) for some \(z_{j}'\).
  Then we have \(u(\vec{z}) = s^{k+1}(u'(z_{1}, \dots, z_{j-1}, z_{j}', z_{j+1}, \dots, z_{l}))\) and a term \(u'\).
  Hence, by Lemma~\ref{lem:3:1} we have \(u(\vec{z}) \neq \numeral{k}\).
\end{proof}
Furthermore, we can eliminate simple \(\UD^{-}\) literals at the expense of introducing several positive literals and an existential quantifier.
\begin{lemma}
  \label{lem:12}
  Let \(k \in \mathbb{N}\), then
  \[
    B + \axiomB{1} \vdash z \neq \numeral{k} \leftrightarrow \left( \Exists{z'}{z = s^{k + 1}z'} \vee \bigvee_{i = 0}^{k - 1}z = \numeral{i} \right).
  \]
\end{lemma}
\begin{proof}
  The ``if'' direction is obvious.
  For the ``only if'' direction assume \(z \neq \numeral{k}\) and proceed by \(k\)-fold case analysis on \(z\).
  If \(z = \numeral{i}\) with \(i < k\), then we are done.
  The case where \(z = \numeral{k}\) contradicts the assumption and therefore we are done.
  If \(z = s^{k + 1}z'\) for some \(z'\), then we are done as well.
\end{proof}
The elimination of the \(\UD\) literals from a component \(\chi(x_{1}, \dots, x_{m})\) consists of two majors steps.
In a first step we deal with all the \(\UD\) literals except the simple \(\UD\) literals of the form \(x_{i} = \numeral{k}\) with \(k \in \mathbb{N}\) and \(i \in \{ 1, \dots, m \}\).
In the second step we will deal with the remaining \(\UD\) literals by making use of the observation that the truth value of a literal of the form \(x = \numeral{k}\) with \(k \in \mathbb{N}\) becomes fixed when \(x\) is large enough.

Let us start by defining some measures that will be used to control the first step of the elimination procedure.
\begin{definition}
  \label{def:32}
  Let \(\chi(\vec{x}) = \Exists{y_{1}}{\dots\Exists{y_{l}}{C_{\chi}}}\) be a component, then \(\#^{-}(\chi)\) is the number of occurrences of negative literals in \(\chi\), \(\#_{\exists}(\chi) = l\), \(\#_{\mathrm{complex}}^{+}(\chi)\) is the number of occurrences of complex \(\UD^{+}\) literals in \(\chi\), and \(\#_{\mathrm{FV}}(\chi)\) is the number of free variables of \(\chi\).
\end{definition}
We will now provide some intermediate lemmas that allow us to eliminate a single literal.
\begin{lemma}[Elimination of \(\UD^{-}\) literals]
  \label{lem:eliminate_single_negative_up_down_literal}
  Let \(\chi(\vec{x})\) be a \(p\)-free component containing a \(\UD^{-}\) literal, then there exist \(p\)-free components \(\chi_{1}'\), \dots, \(\chi_{n}'\) such that \(\BLA + \axiomB{1} \vdash \chi \leftrightarrow \bigvee_{i = 1}^{n}\chi_{i}'\) and \(\#^{-}(\chi_{i}') < \#^{-}(\chi)\) for \(i = 1, \dots, n\).
\end{lemma}
\begin{proof}
  We first apply Lemma~\ref{lem:split_complex_up_down_atom} in order to split the atom of the \(\UD^{-}\) literal.
  After that, we move the negations inwards and apply Lemma~\ref{lem:12} to all the newly introduced literals of the form \(z \neq \numeral{k}\) with \(k \in \mathbb{N}\).
  Now we move the newly introduced existential quantifiers outwards and possibly rename some bound variables.
  Finally, we distribute conjunctions over disjunctions exhaustively.
  Let \(\chi_{1}, \dots, \chi_{k}\) be the resulting components.
  Since we have introduced only existential quantifiers and positive literals, we have \(\#^{-}(\chi_{i}) < \#^{-}(\chi)\).
\end{proof}
\begin{lemma}[Elimination of complex \(\UD^{+}\) literals]
  \label{lem:eliminate_single_positive_complex_up_down_literal}
  Let \(\chi(\vec{x})\) be a \(p\)-free component containing a complex \(\UD^{+}\) literal, then there exist \(p\)-free components \(\chi_{1}'\), \dots, \(\chi_{n}'\) with \(\BLA + \axiomB{1} \vdash \chi \leftrightarrow \bigvee_{i = 1}^{n}\chi_{i}'\) such that \(\#^{-}(\chi_{i}') = \#^{-}(\chi)\), \(\#_{\exists}(\chi_{i}') = \#_{\exists}(\chi)\), and \(\#_{\complex}^{+}(\chi_{i}') < \#_{\complex}^{+}(\chi)\) for \(i = 1, \dots, n\).
\end{lemma}
\begin{proof}
  We apply Lemma~\ref{lem:split_complex_up_down_atom} to split a complex \(\UD^{+}\) literal.
  Now obtain components \(\chi_{1}', \dots, \chi_{n}'\) by distributing conjunctions over disjunctions exhaustively and moving the existential quantifiers inwards over the disjunctions.
  Clearly we have \(\BLA + \axiomB{1} \vdash \chi \leftrightarrow \bigvee_{i = 1}^{n}\chi_{i}'\).
  Observe that this operation does not introduce any negative literals or quantifiers.
  Hence we have \(\#^{-}(\chi_{i}') = \#(\chi)\) and \(\#_{\exists}(\chi_{i}') = \#_{\exists}(\chi)\) for \(i = 1, \dots, n\).
  Moreover, only simple \(\UD^{+}\) literals are introduced, hence we have \(\#_{\complex}^{+}(\chi_{i}') < \#_{\complex}^{+}(\chi)\) for \(i = 1, \dots, n\).
\end{proof}
\begin{lemma}[Elimination of simple \(\UD^{+}\) literals]
  \label{eliminate_single_simple_positive_exists_down_literal}
  Let \(\chi(\vec{x}) = \Exists{y_{1}}{\dots\Exists{y_{l}}{C_{\chi}}}\) be a \(p\)-free component containing a literal of the form \(y_{i} = \numeral{k}\), then there exists a \(p\)-free component \(\chi'(x)\) such that \(\vdash \chi \leftrightarrow \chi'\), \(\#^{-}(\chi') = \#^{-}(\chi)\), and \(\#_{\exists}(\chi') < \#_{\exists}(\chi)\).
\end{lemma}
\begin{proof}
  Let us assume without loss of generality that \(C = y_{i} = \numeral{k} \wedge C'(\vec{x},y_{1}, \dots, y_{l})\), where \(C'\) is a conjunction of literals.
  Then, it suffices to apply the first-order equivalence
  \begin{multline*}
    \vdash \Exists{y_{i}}{\left( y_{i} = \numeral{k} \wedge C'(\vec{x},y_{1}, \dots, y_{i-1},y_{i},y_{i+1}, \dots, y_{l}) \right)} \\ \leftrightarrow C'(\vec{x},y_{1},\dots,y_{i-1},\numeral{k},y_{i+1},\dots, y_{l}).
  \end{multline*}
  Clearly we have \(\#^{-}(\chi') = \#^{-}(\chi)\) and \(\#_{\exists}(\chi') < \#_{\exists}(\chi)\).
\end{proof}
The following lemma combines the previous lemmas in order to accomplish the first step of the elimination of the \(\UD\) literals.
\begin{lemma}
  \label{lem:elimination_up_down_literals_step_1}
  Over \(\BLA + \axiomB{1} + \mathcal{V}\) every \(\exists_{1}(\LLA)\) formula \(\varphi(x_{1},\dots,x_{n})\) is equivalent to a disjunction of formulas of the form \(\bigwedge_{i \in I}x_{i} = \numeral{k_{i}} \wedge \Exists{\vec{y}}{C(\vec{x},\vec{y})}\), where \(I \subseteq [n] = \{ 1, \dots, n \}\) and \(C\) is a \(p\)-free \(0\)-free conjunction of literals that contains only those variables \(x_{i}\) such that \(i \notin I\).
\end{lemma}
\begin{proof}
  Let \(\chi(\vec{x})\) be a \(p\)-free component, then we proceed by induction on the lexicographic order \(\prec\) on \(\mathbb{N}^{4}\) induced by \(\leq\) and show that over \(\BLA + \axiomB{1}\) the component \(\chi\) is equivalent to disjunction of formulas of the form \(\bigwedge_{i \in I}x_{i} = \numeral{k_{i}} \wedge \Exists{\vec{y}}{C(\vec{x},\vec{y})}\), where \(I \subseteq [n]\) and \(C\) is a \(p\)-free disjunction of \(\UU\) and \(\DD\) literals that contains only those variables \(x_{i}\) such that \(i \notin I\).
  Let \(\#(\chi) = (\#^{-}(\chi), \#_{\exists}(\chi), \#_{\complex}^{+}(\chi), \#_{\mathrm{FV}}(\chi))\).
  If \(\chi\) contains a \(\UD^{-}\) literal, then we apply Lemma~\ref{lem:eliminate_single_negative_up_down_literal} in order to obtain \(p\)-free components \(\chi_{1}'\), \dots, \(\chi_{n}'\) such that \(\BLA + \axiomB{1} \vdash \chi \leftrightarrow \bigvee_{i = 1}^{n}\chi_{i}'\) and \(\#^{-}(\chi_{i}') < \#^{-}(\chi)\).
  Hence \(\#(\chi_{i}') \prec \#(\chi)\) and therefore we can apply the induction hypothesis to each of \(\chi_{1}'\), \dots, \(\chi_{n}'\) in order to obtain the desired components.
  If \(\chi\) contains a complex \(\UD^{+}\) literal, then we apply Lemma~\ref{lem:eliminate_single_positive_complex_up_down_literal} in order to obtain \(p\)-free components \(\chi_{1}'\), \dots, \(\chi_{n}'\) with \(\BLA + \axiomB{1} \vdash \chi \leftrightarrow \bigvee_{i = 1}^{n}\chi_{i}'\) such that \(\#^{-}(\chi_{i}') = \#^{-}(\chi)\), \(\#_{\exists}(\chi_{i}') = \#_{\exists}(\chi)\) and \(\#_{\complex}^{+}(\chi_{i}') < \#_{\complex}^{+}(\chi)\), for \(i = 1, \dots, n\).
  Hence \(\#(\chi_{i}') \prec \#(\chi)\) for \(i = 1, \dots, n\) and therefore we can apply the induction hypothesis to \(\chi_{1}'\), \dots, \(\chi_{n}'\) in order to obtain the desired components.
  Let \(\chi(x) = \Exists{y_{1}}{\dots\Exists{y_{l}}{C_{\chi}}}\).
  If \(\chi\) contains a \(\UD\) literal of the form \(x_{i} = \numeral{k_{i}}\) with \(i \in \{ 1, \dots, n\}\), then let \(\chi = \Exists{\vec{y}}{C_{\chi}(\vec{x},\vec{y})}\) and \(\chi' = \Exists{\vec{y}}{C_{\chi}[x_{i}/\numeral{k_{i}}]}\).
  We have \(\vdash \chi \leftrightarrow x_{i} = \numeral{k_{i}} \wedge \chi'\).
  Clearly, \(\#^{-}(\chi') = \#^{-}(\chi)\), \(\#_{\exists}(\chi') = \#_{\exists}(\chi)\), and \(\#_{\complex}^{+}(\chi') \leq \#_{\complex}^{+}(\chi)\) but \(\#_{\mathrm{FV}}(\chi') = \#_{\mathrm{FV}}(\chi) - 1\).
  Hence, we may apply the induction hypothesis to the component \(\chi'\).
  If \(\chi\) contains a simple \(\UD^{+}\) literal \(y_{i} = \numeral{k}\), then we apply Lemma~\ref{eliminate_single_simple_positive_exists_down_literal} in order to obtain a \(p\)-free component \(\chi'(x)\) such that \(\BLA + \axiomB{1} \vdash \chi \leftrightarrow \chi'\), \(\#^{-}(\chi') = \#^{-}(\chi)\), and \(\#_{\exists}(\chi') < \#_{\exists}(\chi)\).
  Hence we have \(\#(\chi') \prec \#(\chi)\) and therefore we can apply the induction hypothesis in order to obtain the desired components.

  Now let \(\varphi(x_{1}, \dots, x_{n})\) be an \(\exists_{1}(\LLA)\) formula.
  By Lemma~\ref{lem:exists_one_formula_to_p_free_components} the formula \(\varphi\) is equivalent over \(\BLA + \axiomB{1}\) to a disjunction of \(p\)-free components.
  Therefore, by the procedure above the formula \(\varphi\) is equivalent over \(\BLA + \axiomB{1}\) to a disjunction of formulas of the form \(\bigwedge_{i \in I}x_{i} = \numeral{k_{i}} \wedge \Exists{\vec{y}}{C(\vec{x},\vec{y})}\), where \(I \subseteq [n]\) and \(C\) is a \(p\)-free disjunction of \(\UU\) and \(\DD\) literals containing only those variables \(x_{i}\) such that \(i \notin I\).
  Now we apply Lemma~\ref{lem:34} to eliminate the \(\DD\) literals from \(C\) and Lemma~\ref{lem:28} to eliminate \(0\) from the \(\UU\) literals of \(C\).
\end{proof}
In the next step we eliminate the remaining literals of the form \(x = \numeral{k}\).
This step relies on the observation that the truth value of these literals is fixed when \(x\) is large enough.
\begin{proposition}
  \label{pro:13}
  \label{pro:linear_arithmetic:exists_one_p_free_zero_free_via_inflation}
  Let \(\varphi(x_{1},\dots,x_{n})\) be an \(\exists_{1}\) formula, then there exists \(N \in \mathbb{N}\) such that \(\varphi(s^{N}(x_{1}), \dots, s^{N}(x_{n}))\) is equivalent over \(\BLA + \axiomB{1} + \mathcal{V}\) to a \(0\)-free, \(p\)-free, \(\exists_{1}\) formula.
\end{proposition}
\begin{proof}
  By Lemma~\ref{lem:elimination_up_down_literals_step_1} the formula \(\varphi\) is equivalent over \(\BLA + \axiomB{1} + \mathcal{V}\) to a disjunction of the form
  \[
    \bigvee_{j = 1}^{m}\left( \bigwedge_{i \in I_{j}} x_{i} = \numeral{k}_{j,i} \wedge \Exists{\vec{y}}{C_{j}(\vec{x},\vec{y})}\right),
  \]
  where for \(j = 1, \dots, m\), \(I_{j} \subseteq [n]\) and \(C_{j}\) is a \(p\)-free \(0\)-free disjunction of literals containing only those variables \(x_{i}\) such that \(i \notin I_{j}\).
  Let \(N = 1 + \max \{ k_{i,j} \mid j = 1, \dots, m, i \in I_{j}\}\), then \(\varphi(s^{N}(x_{1}), \dots, s^{N}(x_{n}))\) is equivalent over \(\BLA + \axiomB{1} + \mathcal{V}\) to the formula
  \[
    \bigvee_{
      \substack{j = 1, \dots, m
      \\
      I_{j} = \varnothing
    }}
    \Exists{\vec{y}}{C_{j}(\vec{x},\vec{y})}.
  \]
  Finally, we obtain the desired formula by moving the \(\exists\) quantifiers outwards over the disjunction.
  \end{proof}
\subsection{Components in \(\mathbb{N}\) and \(\mathbb{Z}\)}
\label{sec:linear-systems}
In this section we will investigate some basic model-theoretic properties of \(\exists_{1}\) formulas in the structures \(\mathbb{N}\) and \(\mathbb{Z}\).
\begin{definition}
  \label{def:28}
  Let \(M\) be an \(\LLA\) structure and \(\varphi(x)\) a formula. We say that \(d \in M\) is a solution of \(\varphi\) in \(M\) if \(M \models \varphi(d)\).
\end{definition}
We will show that every \(p\)-free \(\exists_{1}\) formula with enough solutions in \(\mathbb{N}|_{\{0,s,+\}}\) has an infinite strictly descending sequence of solutions in \(\mathbb{Z}|_{\{0,s,+\}}\).
Since the structure \(\mathbb{N}|_{\{0,s,+\}}\) can be embedded into \(\mathbb{Z}|_{\{0,s,+\}}\), it is clear that if \(\mathbb{N}|_{\{0,s,+\}} \models \varphi(n)\), then \(\mathbb{Z}|_{\{0,s,+\}} \models \varphi(n)\), for all \(n \in \mathbb{N}\).

Let \(\theta(x_{1}, \dots, x_{k})\) be an atom, then it is obvious that \(\theta\) is equivalent in \(\mathbb{Z}\) to a linear equation of the form \(\sum_{i = 1}^{k}a_{i}x_{i} = b\) with integers \(a_{1}, \dots, a_{k}, b\).
Hence a conjunction of atoms \(\theta_{1}(x_{1}, \dots, x_{k}), \dots, \theta_{n}(x_{1}, \dots,x_{k})\) is equivalent over \(\mathbb{Z}\) to an inhomogeneous system of linear equations of the form
\begin{equation}
  \label{eq:2}
  A\vec{x}  = b, \tag{I}
\end{equation}
where \(A \in \mathbb{Z}^{m \times k}\) and \(b \in \mathbb{Z}^{m \times 1}\).
Now consider the corresponding homogeneous system
\begin{equation}
  \label{eq:8}
  A \vec{x} = \mathbf{0}. \tag{H}
\end{equation}
The solutions of the system \eqref{eq:8} form a submonoid \(\mathcal{H}\) of \(\mathbb{Z}^{k}\) with pointwise addition.
Furthermore, assume that \eqref{eq:2} has a particular solution \(i_{(p)}\), then the set of solutions of \eqref{eq:2} is given by \[
  \mathcal{I} = \{ h + i_{(p)} \mid h \in \mathcal{H} \}.
\]
\begin{lemma}
  \label{lem:17}
  Let \(\chi(x)\) be a component with two solutions in \(\mathbb{Z}\), then for all \(n \in \mathbb{N}\) there exists \(n' \in \mathbb{N}\) with \(n' \geq n\) such that \(\mathbb{Z} \models \chi(-n')\).
\end{lemma}
\begin{proof}
  Let \(\chi(x_{0}) = \Exists{x_{1}}{\dots\Exists{x_{k}}{C_{\chi}(\vec{x}^{T})}}\) with \(\vec{x}^{T} = (x_{0}, x_{1}, \dots, x_{k})\).
  Let \(\mathbb{Z} \models C_{\chi}(n_{(i)}^{T})\) with \(n_{(i)}^{T} = (n_{i,0}, \dots, n_{i,k})\) and \(i \in \{ 1, 2 \}\) such that \(n_{1,0} < n_{2,0}\).

  We start by considering the positive literals of \(\chi\).
  By the above the positive literals of \(\chi\) are equivalent in \(\mathbb{Z}\) to an inhomogeneous linear system 
  \begin{equation}
    \label{eq:I}
    A\vec{x} = b, \tag{I}
  \end{equation}
  with \(A \in \mathbb{Z}^{l\times(k+1)}\) and \(b \in \mathbb{Z}^{l \times 1}\), where \(l\) is the number of positive literals of \(\chi\).
  Let us denote by \(\mathcal{I}\) the set of solutions of \(\eqref{eq:I}\) and by \(\mathcal{H}\) the set of solutions of the homogeneous system.
  Then \(n_{(1)}, n_{(2)} \in \mathcal{I}\) and therefore \(h_{0} \coloneqq n_{(1)} - n_{(2)} \in \mathcal{H}\).
  Hence \(m \cdot h_{0} + n_{(1)} \in \mathcal{I}\) for all \(m \in \mathbb{N}\).

  Now consider a negative literal \(p(x_{0}, \dots, x_{k}) \neq 0\) of \(\chi\), where \(p\) is a linear polynomial in the variables \(x_{0}, \dots, x_{k}\) with coefficients in \(\mathbb{Z}\).
  Let \(q(m) \coloneqq p( (m \cdot h_{0} + n_{(1)})^{T})\), then \(q\) is a linear polynomial in one variable and moreover by the assumptions we have \(q(0) = p(n_{(1)}^{T}) \neq 0\).
  Hence, there clearly is at most one \(k \in \mathbb{Z}\) such that \(q(k) = 0\).

  Let \(p_{1}(x_{0}, \dots, x_{k}) \neq 0\), \dots, \(p_{r}(x_{0}, \dots, x_{k}) \neq 0\) be all the negative literals of \(\chi\), then we let
  \[
    m_{0} = \max \left( \{ 0 \} \cup \bigcup_{i = 1}^{r}\left \{ m + 1 \mid m \in \mathbb{N}, q_{i}(m) = 0 \right \} \right).
  \]
  Clearly, the natural number \(m_{0}\) exists and we have \(\mathbb{Z} \models C_{\chi}((m \cdot h_{0} + n_{(1)})^{T})\) for all \(m \in \mathbb{N}\) with \(m \geq m_{0}\).
  Since \(n_{1,0} < n_{2,0}\) we are done.
\end{proof}
We summarize the results of this section in the following proposition.
\begin{proposition}
  \label{pro:2}
  Let \(\varphi(x)\) be a \(p\)-free \(\exists_{1}\) formula.
  There exists an \(n \in \mathbb{N}\) such that if \(\varphi\) has \(n\) solutions in \(\mathbb{N}\), then there exists an infinite strictly descending sequence of integers \((k_{i})_{i \in \mathbb{N}}\) with \(\mathbb{Z} \models \varphi(k_{i})\) for all \(i \in \mathbb{N}\).  
\end{proposition}
\begin{proof}
  Let \(\chi_{1}, \dots, \chi_{k}\) be \(p\)-free components such that \(\vdash \varphi \leftrightarrow \bigvee_{i = 1}^{k}\chi_{i}\).
  Let \(n = k + 1\) and assume that \(\varphi\) has \(n\) solutions in \(\mathbb{N}\).
  Then by the pigeonhole principle there is a component \(\chi_{i_{0}}\) with two solutions in \(\mathbb{N}\) and therefore \(\chi_{i_{0}}\) has two solutions in \(\mathbb{Z}\).
  Finally, we apply Lemma~\ref{lem:17} to \(\chi_{i_{0}}\).
\end{proof}
\subsection{A non-standard model}
\label{sec:model-theor-constr}
In this section we construct a family of non-standard structures for the language \(\LLA\) and we make use of the results from \cref{sec:syntatic_simplifications,sec:linear-systems} in order to show that these structures are models of the theory
\[
  (\BLA + \axiomB{2} + \axiomB{3}) + \RINDParameterFree{\exists_{1}(\LLA)}.
\]

Let us start by introducing some terminology about the models of this theory.
Since already the theory \([\BLA, \RINDParameterFree{\Open(\LLA)}]\) proves \(\axiomB{1}\) and \(\theoryV\), the models of \((\BLA + \axiomB{2} + \axiomB{3}) + \RINDParameterFree{\exists_{1}(\LLA)}\) are composed of a copy of the natural numbers---the standard elements---and copies of the integers, which we call the non-standard elements.
The elements of the models we construct below are pairs of the form \(n^{[m]} = (m,n) \in \mathbb{N} \times \mathbb{Z}\) such that \(m = 0\) implies \(n \in \mathbb{N}\).
If \(m = 0\), then the element is a standard element, otherwise it is non-standard and belongs to the \(m\)-th copy of the integers.
We call \(m\) the type of the element and \(n\) the value of the element.
We start by defining an operation that will allow us to relate the types of the elements.
\begin{definition}
  \label{def:29}
  The function \(\upharpoonleft : \mathbb{N} \times \mathbb{N} \to \mathbb{N}\) is given by
  \[
    n \upharpoonleft m \coloneqq
  \begin{cases}
    n & \text{if \(n \neq 0\)} \\
    m & \text{if \(n = 0\)}
    \end{cases}.
  \]
\end{definition}
\begin{definition}
  \label{def:3}
  Let \(I \in \mathbb{N}\), then the \(\LLA\) structure \(M_{I}\) consist of pairs of the form \(n^{[m]}\) with \(n \in \mathbb{Z}\), \(m \in \mathbb{N}\) and \(m \leq I\) such that if \(m = 0\) then \(n \in \mathbb{N}\).
  Furthermore, we let \(M_{I}\) interpret the non-logical symbols as follows
  \begin{gather*}
    0^{M_{I}} = 0^{[0]}, \\
    \text{\(s^{M_{I}}(n^{[m]}) = (n + 1)^{[m]}\), for  \(n^{[m]} \in M\)},\\
    p^{M_{I}}(n^{[m]}) =
    \begin{cases}
      (n \dotminus 1)^{[0]} & \text{if \(m = 0\),} \\
      (n - 1)^{[m]} & \text{otherwise}.
    \end{cases}
    \\
    n_{1}^{[m_{1}]} +^{M_{I}} n_{2}^{[m_{2}]} = (n_{1} + n_{2})^{[m_{1} \upharpoonleft m_{2}]}.
  \end{gather*}
\end{definition}
The structure \(M_{0}\) is isomorphic to the standard model \(\mathbb{N}\).
Since we are interested in constructing non-standard structures, we will consider mainly the structures \(M_{I}\) with \(I \geq 1\).
%
%
\begin{lemma}
  \label{lem:2}
  Let \(I \in \mathbb{N}\), then \(M_{I} \models \BLA + \axiomB{1} + \axiomB{3} + \theoryV\).
\end{lemma}
\begin{proof}
  Routine.
\end{proof}
The structures \(M_{I}\) with \(I \in \mathbb{N}\) and \(I > 0\) have the crucial property that \(\mathbb{Z}|_{\{ s, p , + \}}\) can be embedded into the non-standard parts of \(M_{I}\).
Hence the solutions of \(0\)-free \(\exists_{1}\) formulas in \(\mathbb{Z}\) can be transferred into the non-standard chains of \(M_{I}\).
\begin{lemma}
  \label{lem:30}
  Let \(m, I \in \mathbb{N}\) with \( 1 \leq m \leq I\), then the function \(\iota_{m}: \mathbb{Z} \rightarrow M_{I}\), \(n \mapsto n^{[m]}\) is an embedding of \(\mathbb{Z}|_{\{ s, p , + \}}\) into \(M_{I}|_{\{ s, p , + \}}\).
  In particular, if \(\varphi(x)\) is a \(0\)-free \(\exists_{1}\) formula, then \(\mathbb{Z} \models \varphi(n)\) implies \(M_{I} \models \varphi(n^{[m]})\).
\end{lemma}
\begin{proof}
  It is routine to verify that \(\iota_{m}\) is an embedding of \(\mathbb{Z}|_{\{ s, p , + \}}\) into \(M_{I}|_{\{ s, p , + \}}\).
  The second part of the claim follows from the well-known fact that embeddings preserve \(\exists_{1}\)  formulas (see for example  \cite[Theorem~2.4.1]{hodges1997})
\end{proof}
We can now show that the structures \(M_{I}\) satisfy unnested applications of the induction rule \(\RINDParameterFree{\exists_{1}(\LLA)}\).
\begin{theorem}
  \label{thm:2}
  Let \(I \geq 1\) and \(T\) be a sound \(\LLA\) theory such that \(M_{I} \models T\), then \(M_{I} \models [T,\RINDParameterFree{\exists_{1}(\LLA)}]\).
\end{theorem}
\begin{proof}
  Let \(\varphi(x)\) be an \(\exists_{1}\) formula and assume that \(\varphi\) is \(T\)-inductive.
  Since \(T\) is sound, we have \(\mathbb{N} \models \varphi(x)\).
  Now consider an element \(n^{[m]} \in M\).
  If \(m = 0\), then \(n \in \mathbb{N}\) and by the observation above \(\mathbb{N} \models \varphi(\numeral{n})\).
  By the \(\exists_{1}\)-completeness of \(\BLA\) we have \(B \vdash \varphi(\numeral{n})\) and therefore \(M_{I} \models \varphi(\numeral{n})\).
  Since \(M_{I} \models \numeral{n} = n^{[0]}\) we obtain \(M_{I} \models \varphi(n^{[m]})\).

  Now assume that \(m > 0\).
  By Proposition~\ref{pro:13} there exists a \(0\)-free \(p\)-free formula \(\varphi'\) and an \(N \in \mathbb{N}\) such that \(B + \axiomB{1} + \theoryV \vdash \varphi(s^{N}(x)) \leftrightarrow \varphi'(x)\).
  Hence we have \(\mathbb{N} \models \varphi'\) and therefore by Proposition~\ref{pro:2} there is an infinite strictly descending sequence of integers \((k_{i})_{i \in \mathbb{N}}\) such that \(\mathbb{Z} \models \varphi'(k_{i})\) for \(i \in \mathbb{N}\).
  By Lemma~\ref{lem:30} we obtain \(M_{I} \models \varphi'(k_{i}^{[m]})\) for \(i \in \mathbb{N}\).
  Hence there exists \(i_{0} \in \mathbb{N}\) such that \(k_{i_{0}} + N \leq n\) and \(M_{I} \models \varphi'(k_{i_{0}}^{[m]})\), thus, \(M_{I} \models \varphi((k_{i_{0}} + N)^{[m]})\).
  Since \(M_{I} \models \varphi(x) \rightarrow \varphi(s(x))\), we obtain \(M_{I} \models \varphi(k_{i_{0}} + N + k)^{[m]}\) for all \(k \in \mathbb{N}\).
  In particular, we have \(M_{I} \models \varphi(n^{[m]})\).
\end{proof}
By iterating the argument above we can show that the structures \(M_{I}\) even satisfy nested applications of \(\RINDParameterFree{\exists_{1}(\LLA)}\).
\begin{corollary}
  \label{cor:5}
  Let \(I \geq 1\) and \(T\) be a sound \(\LLA\) theory such that \(M_{I} \models T\), then \(M_{I} \models T + \RINDParameterFree{\exists_{1}(\LLA)}\).
\end{corollary}
\begin{proof}
  We proceed by induction on \(j\) to show \([T, \RINDParameterFree{\exists_{1}(\LLA)}]_{j}\).
  For the base case we have \([T,\RINDParameterFree{\exists_{1}(\LLA)}]_{0} = T\), hence we are done.
  For the induction step assume that \(M_{I} \models [T,\RINDParameterFree{\exists_{1}(\LLA)}]_{j}\) and observe that \([T,\RINDParameterFree{\exists_{1}(\LLA)}]_{j}\) is sound.
  Now obtain \(M_{I} \models [T,\RINDParameterFree{\exists_{1}(\LLA)}]_{j+1}\) by Theorem~\ref{thm:2}.
\end{proof}
\begin{lemma}
  \label{lem:4}
  \(M_{1} \models \axiomB{2}\).
\end{lemma}
\begin{proof}
  Clearly it suffices to show that \(b_{1} \upharpoonleft b_{2} = b_{2} \upharpoonleft b_{1}\) for \(b_{1}, b_{2} \in \{ 0 , 1 \}\).
  The only interesting case is \(b_{1} \neq b_{2}\), that is, \(0 \upharpoonleft 1 = 1 = 1 \upharpoonleft 0\).
\end{proof}
\begin{corollary}
  \label{cor:4}
  \(M_{1} \models (\BLA + \axiomB{2} + \axiomB{3}) + \RINDParameterFree{\exists_{1}(\LLA)}\).
\end{corollary}
\begin{proof}
  An immediate consequence of Lemma~\ref{lem:4} and Corollary~\ref{cor:5}.
\end{proof}
Theorem~\ref{thm:5} can now finally be obtained as an immediate consequence of Corollary~\ref{cor:4}.
\begin{proof}[Proof of Theorem~\ref{thm:5}]
  By Corollary~\ref{cor:4} we can work with \(M_{1}\).
  Now observe that \(n \cdot k^{[1]} + \numeral{(m - n)k} = (nk)^{[1]} + ((m-n)k)^{[0]} = (nk + (m -n)k)^{[1]} = (mk)^{[1]} = m \cdot k^{[1]}\), but \(k^{[1]} \neq k^{[0]}\).
\end{proof}
Let us now consider whether some straightforward modifications of the background theory in Theorem~\ref{thm:5} will improve the result.
The following lemma shows that we do not strengthen the result of Theorem~\ref{thm:5} by adding any \(\forall_{1}\) consequence of \((\BLA + \axiomB{2} + \axiomB{3}) + \RINDParameterFree{\exists_{1}(\LLA)}\) to the background theory.
\begin{lemma}
  \label{lem:20}
  Let \(L\) be a first-order language, \(T\) an \(L\), \(\forall_{1}\) theory and \(U\) a theory, then
  \[
    [T + U, \RINDParameterFree{\exists_{1}(L)}]_{n} \equiv T + [U, \RINDParameterFree{\exists_{1}(L)}]_{n}, \ \text{for all \(n \in \mathbb{N}\)}.
  \]
  Furthermore, \((T + U) + \RINDParameterFree{\exists_{1}(L)} \equiv T + (U + \RINDParameterFree{\exists_{1}(L)})\).
\end{lemma}
\begin{proof}
  The first part is obtained by a straightforward induction on \(n\) and applying Lemma~\ref{lem:10}.
  The second part is an immediate consequence of the first part.
\end{proof}
%

%
%
The next natural question to ask is whether removing the formulas \(\axiomB{2}\) and \(\axiomB{3}\) in Theorem~\ref{thm:5} would weaken the result.
The following result shows that removing the axiom \(\axiomB{2}\) would indeed weaken the result.
\begin{lemma}
  \label{lem:19}
  \((\BLA + \axiomB{3}) + \RINDParameterFree{\exists_{1}(\LLA)} \not \vdash \axiomB{2}\).
\end{lemma}
\begin{proof}
  By Lemma~\ref{lem:2} we have \(M_{2} \models \BLA + \axiomB{3}\).
  Since \(0^{[1]} + 0^{[2]} = 0^{[1 \upharpoonleft 2]} = 0^{[1]} \neq 0^{[2]} = 0^{[2 \upharpoonleft 1]} = 0^{[2]} + 0^{[1]}\), we obtain \(M_{2} \models (\BLA + \axiomB{3}) + \RINDParameterFree{\exists_{1}(\LLA)} \not \vdash \axiomB{2}\) by Corollary~\ref{cor:5}.
\end{proof}
%
%
We conjecture that removing \(\axiomB{3}\) from Theorem~\ref{thm:5} would weaken the result as well.
\begin{conjecture}
  \((\BLA + \axiomB{2}) + \RINDParameterFree{\exists_{1}(\LLA)} \not \vdash \axiomB{3}\).
\end{conjecture}

%
%
We conclude this section by observing that the model-theoretic construction developed in this section does not yield a proof of Conjecture~\ref{con:2}.
\begin{lemma}
  \label{lem:15}
  Let \(I \geq 1\), then \(M_{I} \not \models \INDParameterFree{\exists_{1}(\LLA)}\).
\end{lemma}
\begin{proof}
  Let \(\chi(x) = \Exists{y_{1}}{\Exists{y_{2}}{\Exists{y_{3}}{\theta(x, y_{1}, y_{2}, y_{3})}}}\) with
  \[
    \theta(x, y_{1}, y_{2}, y_{3}) \coloneqq x + y_{1} \neq x + y_{2} \wedge x + (y_{3} + y_{1}) = x + (y_{3} + y_{2}).
  \]
  We will show that \(M_{I} \not \models I_{x} \chi(x)\).
  We first show that \(M_{I} \models \chi(n^{[0]})\) for all \(n \in \mathbb{N}\).
  For this it suffices to observe \(n^{[0]} + 0^{[0]} = n^{[0]}\), \(n^{[0]} + 0^{[1]} = n^{[1]}\) and \(n^{[0]} + (0^{[1]} + 0^{[0]}) = n^{[0]} + 0^{[1]} = n^{[1]} = n^{[0]} + (0^{[1]} + 0^{[1]})\).
  Hence \(M_{I} \models \theta(n^{[0]}, 0^{[0]}, 0^{[1]}, 0^{[1]})\).
  Now we will show that \(M_{I} \not \models \chi(n^{[m]})\) for \(n \in \mathbb{Z}\) and \(m > 0\).
  Let \(k_{1}^{[m_{1}]}\), \(k_{2}^{[m_{2}]}\), \(l^{[h]} \in M_{I}\) and assume that
  \begin{gather}
    n^{[m]} + k_{1}^{[m_{1}]} \neq n^{[m]} + k_{2}^{[m_{2}]},  \label{eq:5} \tag{*} \\
    n^{[m]} + (l^{[h]} + k_{1}^{[m_{1}]}) = n^{[m]} + (l^{[h]} + k_{2}^{[m_{2}]}). \label{eq:6} \tag{\(\dagger\)}
  \end{gather}
  Since \(m > 0\), we have \(m \upharpoonleft u = m\) for all \(u \in \mathbb{Z}\).
  Hence by \eqref{eq:5} we obtain \(n + k_{1} \neq n + k_{2}\).
  By \eqref{eq:6} we obtain \(n + l + k_{1} = n + l + k_{2}\), thus \(k_{1} = k_{2}\).
  Therefore \(n + k_{1} = n + k_{2}\).
  Contradiction!

  By the above we thus have \(M_{I} \models \chi(0)\).
  Now let \(n^{[m]} \in M_{I}\). If \(m = 0\), then we have \(M_{I} \models \chi((n + 1)^{[0]})\) hence \(M_{I} \models \chi(n^{[0]}) \rightarrow \chi((n+1)^{[0]})\).
  If \(m > 0\), then we have \(M_{I} \not \models \chi(n^{[m]})\) hence \(M_{I} \models \chi(n^{[m]}) \rightarrow \chi((n + 1)^{[m]})\).
  Thus \(M_{I} \not \models I_{x}\chi(x)\).
\end{proof}
On the other hand it may be interesting to observe that already unnested application of the parameter-free induction rule for \(\exists_{1}\) formulas contain the parameter-free induction schema for quantifier-free formulas.
\begin{lemma}
  \label{lem:9}
  Let \(L\) be a language, then \([\varnothing, \RINDParameterFree{\exists_{1}(L)}] \vdash \INDParameterFree{\Open(L)}\).
\end{lemma}
\begin{proof}
  Let \(\varphi(x)\) be a quantifier-free formula.
  Let \(\psi(x,y)\) be given by
  \[
    \left(\varphi(0) \wedge (\varphi(y) \rightarrow \varphi(s(y)) \right) \rightarrow \varphi(x).
  \]
  By shifting the existential quantifier inward, it is straightforward to see that \(\vdash \Forall{x}{\Exists{y}{\psi(x,y)}} \leftrightarrow I_{x}\varphi\).
  Moreover, we clearly have \(\vdash \Exists{y}{\psi(0,y)}\).
  Now work in \(\varnothing\) and assume \(\psi(x,y_{0})\).
  Assume furthermore \(\varphi(0)\), \(\varphi(y_{0}) \rightarrow \varphi(s(y_{0}))\) and \(\varphi(x) \rightarrow \varphi(s(x))\).
  By the assumptions we obtain \(\varphi(x)\) and moreover \(\varphi(s(x))\).
  Hence, we have \(\Exists{y}{\psi(s(x),y)}\).
  Therefore, \(\vdash \Exists{y}{\psi(x,y)} \rightarrow \Exists{y}{\psi(s(x),y)}\).
  Thus, \([\varnothing, \RINDParameterFree{\exists_{1}(L)}] \vdash \Forall{x}{\Exists{y}{\psi(x,y)}}\).
\end{proof}
\section{Conclusion}
\label{sec:conclusion}
Clause set cycles are a formalism introduced by the authors of this article in \cite{hetzl2020} for the purpose of giving an upper bound on the strength on a family of \ac{AITP} systems based on the extension of a saturation theorem prover by a cycle detection mechanism, such as the n-clause calculus \cite{kersani2013,kersani2014}.
In this article we have extended the analysis of clause set cycles that was begun in \cite{hetzl2020} by providing a logical characterization of refutation by a clause set cycle and concrete clause sets that are not refutable by a clause set cycle but that are refuted by induction for quantifier-free formulas.

In Section~\ref{sec:clause-set-cycles} we have identified several logical features of clause set cycles.
Identifying these features has enabled us to give a characterization of the notion of refutation by a clause set cycle in terms of a logical theory.
The characterization allows us to think of clause set cycles essentially as unnested applications of the parameter-free \(\exists_{1}\) \(\eta\) induction rule.
In the light of this logical characterization we were able to reduce the task of finding clause sets that are not refuted by a clause set cycle to an independence problem.

Based on this characterization we have shown two unprovability results for clause set cycles.
The first result (Corollary~\ref{cor:3}) exploits the fact that refutations by a clause set cycle only make use of \(\eta\)-instances of the inductive lemmas.
In particular, we have shown that even the full induction schema subject to the \(\eta\)-restriction does not prove some atoms that can already be obtained by an unnested application of the open parameter-free induction rule.
This shows that the \(\eta\)-restriction is very limiting.
However, our second unprovability result (Corollary~\ref{cor:6}) does not rely on the \(\eta\)-restriction and thus shows that \ac{AITP} systems based on clause set cycles have more limitations.
In Section~\ref{sec:idemp-line-arithm} we have developed the underlying independence result (Theorem~\ref{thm:5}).
This independence result shows us that the unprovability persists even when the induction rule is nested.
We conjecture that this unprovability phenomenon is due to the absence of induction parameters and therefore also persists when the induction rule is replaced by the induction schema.
This second unrefutability result shows that clause set cycles fail to capture induction arguments that involve very simple generalizations.

The results in this article together with the results in \cite{hetzl2020} explain much about the situation of \ac{AITP} systems based on clause set cycles in the logical landscape.
We have summarized the current results as well as some conjectures in Figure~\ref{fig:1}.
The figure depicts the refutational strength of various induction systems.
The set of clause sets refuted by a system is described by an arc.
The name of the system is inscribed near the top of the corresponding arc.
The systems range over all first-order languages.
The system \(\systemCSC\) denotes refutation by a clause set cycle and \(\systemNCC\) denotes the n-clause calculus as described in \cite{hetzl2020}.
The points \(\{\bullet_{i} \mid i = 1, 2, 3, 4, 5\}\) represent clause sets whose positions are confirmed by the results in this article and \cite{hetzl2020}.
In particular \(\bullet_{1}\) corresponds to the clause set that witnesses \cite[Corollary~5.8]{hetzl2020}, \(\bullet_{2}\) corresponds to the clause set constructed in Section~\ref{sec:unrefutability_instance_restriction}, and \(\bullet_{3}\) corresponds to the clause sets constructed in Section~\ref{sec:unrefutability_rule_and_parameters}.
the points \(\bullet_{4}\), \(\bullet_{5}\) correspond to some the clause sets mentioned in Theorem~\ref{thm:separation_nesting_of_induction_rule}.
The inclusion of \(\INDParameterFree{\Open}\) in \([\varnothing,\RINDParameterFree{\exists_{1}}]\) is shown by Lemma~\ref{lem:9}.

The dashed arc corresponding to the system \(\INDParameterFree{\exists_{1}}\) is positioned according to Conjecture~\ref{con:2}.
The points \(\{ \ast_{i} \mid i = 1, 2\}\) of Figure~\ref{fig:1} represent clause sets that we conjecture to be at the respective positions.
In particular, the point \(\ast_{1}\) corresponds to the clause set mentioned in Conjecture~\ref{con:axiom_stronger_than_rule}.
We would like to clarify the status of the point \(\ast_{1}\) and the dashed arc corresponding to the system \(\INDParameterFree{\exists_{1}}\), as this would contribute to the understanding of the role of induction parameters and induction rules in automated inductive theorem proving.
%
%



Due to the recent advances in saturation-based theorem proving, the research on automated inductive theorem proving has recently increasingly focused on the integration of induction into saturation-based theorem provers \cite{cruanes2015,cruanes2017,kersani2013,kersani2014,wand2017,echenim2020,reger2019,hajdu2020}.
We plan to carry out similar investigations for all these methods in order to develop a more global and unified view of induction in saturation-based theorem proving.
In particular these investigations will give rise to the analysis of the interaction of the induction principle with various mechanisms of saturation-based provers such as Skolemization, splitting, term orderings, and redundancy criteria.


The point \(\ast_{2}\) in Figure~\ref{fig:1} gives rise to a more general topic that is worth mentioning separately.
On the one hand it is computationally expensive for \ac{AITP} systems to carry out even a small number of inductions, and on the other hand the space of all possible induction formulas is very large.
Hence \ac{AITP} systems rely on heuristics to find induction formulas such as restricting the overall shape of the considered induction formulas and drawing syntactical material for induction from the formulas generated during the proof search.
For example, the n-clause calculus as described in \cite{kersani2013,kersani2014} only makes use of clause set cycles that appear as a subset of the clauses that are generated by the underlying saturation-based system.
Such heuristics will not succeed in cases where a sufficiently non-analytic induction is required.
Our technique for analyzing \ac{AITP} systems as logical theories can deal with such heuristics only to a limited extent.
For example, the notion of refutation by a clause set cycle completely ignores the fact that the n-clause calculus draws clause set cycles only from the generated clauses.
Once the logical strength of most inductive theorem provers is known precisely enough it will likely be necessary to investigate the fine grained analyticity properties of the provers in order to get a better understanding of the consequences of restricting the degree of analyticity.
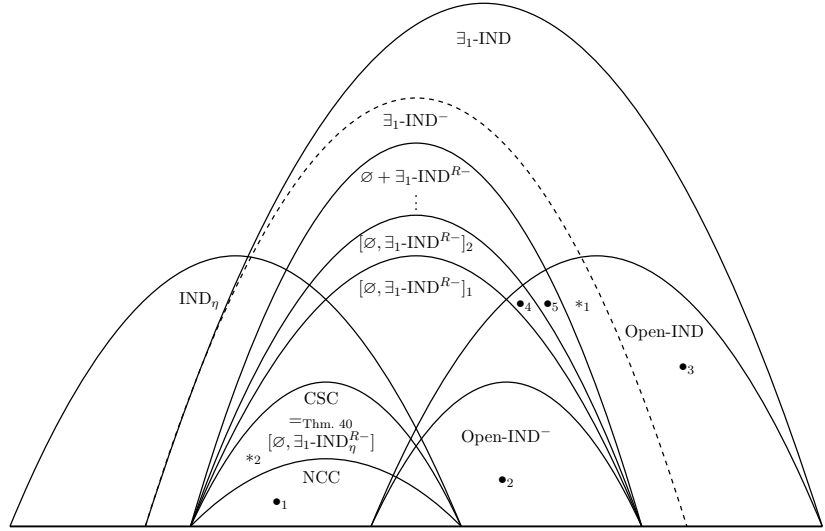
\begin{figure}[ht]
  \centering
  \scalebox{.6}{
    \input{paper_overview_figure}
  }
  \caption{
    Overview of the clausal refutational strength of various induction systems.}
  \label{fig:1}
\end{figure}
\paragraph{Acknowledgments.}
We thank Emil Je{\v{r}}{\'a}bek for pointing out to us the similarity of clause set cycles with unnested induction rules.
Moreover, we thank the anonymous reviewers whose feedback helped to improve the article significantly.

\begin{acronym}
  \acro{AITP}{automated inductive theorem proving}
\end{acronym}

\bibliography{./bibliography/bibliography.bib}
\bibliographystyle{alpha}

\end{document}

%% file: header.tex
\usepackage[utf8]{inputenc}
\usepackage[T1]{fontenc}
\usepackage{authblk}
\usepackage{graphicx}
\usepackage{grffile}
\usepackage{longtable}
\usepackage{wrapfig}
\usepackage{rotating}
\usepackage[normalem]{ulem}
\usepackage{amsmath}
\usepackage{textcomp}
\usepackage{amssymb}
\usepackage{capt-of}
\usepackage{hyperref}
\usepackage{url}
\usepackage[capitalize]{cleveref}

\usepackage{mathtools}

\usepackage[nolist, nohyperlinks]{acronym}
\usepackage{amsmath}
\usepackage{amssymb}
\usepackage{amsthm}
\usepackage{stmaryrd}
\usepackage{todonotes}
\usepackage{ebproof}
\usepackage{pdflscape}
\usepackage{enumitem}

\usepackage{tikz}
\usetikzlibrary{cd,tikzmark,decorations.pathreplacing,calligraphy}

\usepackage{standalone}

\usepackage{fixme}
\fxsetup{layout=footnote}
\fxsetup{status=final}

\usepackage[inline,final]{showlabels}

\usepackage{comment}

\DeclareUrlCommand\emailUrl{\urlstyle{rm}}
\newcommand{\email}[1]{\href{mailto:#1}{\protect\emailUrl{#1}}}
\newcommand{\emailVierling}{\email{jannik.vierling@tuwien.ac.at}}
\newcommand{\emailHetzl}{\email{stefan.hetzl@tuwien.ac.at}}

\newtheorem{theorem}{Theorem}
\newtheorem*{theorem*}{Theorem}
\newtheorem{proposition}[theorem]{Proposition}
\newtheorem{definition}[theorem]{Definition}
\newtheorem{lemma}[theorem]{Lemma}
\newtheorem{corollary}[theorem]{Corollary}
\newtheorem{conjecture}[theorem]{Conjecture}
\newtheorem*{conjecture*}{Conjecture}
\newtheorem{remark}[theorem]{Remark}
\newtheorem{example}[theorem]{Example}

\newcommand{\Open}{\mathrm{Open}}

\newcommand{\IND}[1]{{#1}\text{-}\mathrm{IND}}
\newcommand{\INDParameterFree}[1]{{#1}\text{-}\mathrm{IND}^{-}}
\newcommand{\RIND}[1]{{#1}\text{-}\mathrm{IND}^{R}}
\newcommand{\RINDParameterFree}[1]{{#1}\text{-}\mathrm{IND}^{R-}}
\newcommand{\RINDParameterFreeEta}[1]{{#1}\text{-}\mathrm{IND}_{\eta}^{R-}}

\newcommand{\LanguageOf}[1]{L(#1)}

\newcommand{\cnf}{\mathit{cls}}
\newcommand{\LA}{\mathrm{LA}}
\newcommand{\LLA}{L_{\LA}}
\newcommand{\BLA}{B}
\newcommand{\BLAPrime}{\BLA^{\prime}}
\newcommand{\theoryV}{\mathcal{V}}

\newcommand{\Var}{\mathit{Var}}
\newcommand{\dotminus}{\mathop{\dot{-}}}

\newcommand{\axiomA}[1]{\mathrm{A{#1}}}
\newcommand{\axiomB}[1]{\mathrm{B{#1}}}
\newcommand{\axiomV}[1]{\mathrm{V_{#1}}}

\newcommand{\formulaE}[3]{\mathrm{E}_{{#1},{#2},{#3}}}

\newcommand{\numeral}[1]{\overline{#1}}

\newcommand{\systemCSC}{\mathrm{CSC}}
\newcommand{\systemNCC}{\mathrm{NCC}}

\newcommand{\UU}{{\uparrow}{\uparrow}}
\newcommand{\UD}{{\uparrow}{\downarrow}}
\newcommand{\DD}{{\downarrow}{\downarrow}}

\newcommand{\complex}{\mathrm{complex}}

\newcommand{\QuantifierParenthesis}[3]{({#1}{#2}){#3}}

\newcommand{\Quantifier}[3]{\QuantifierParenthesis{#1}{#2}{#3}}
\newcommand{\Forall}[2]{\Quantifier{\forall}{#1}{#2}}

\newcommand{\Exists}[2]{\Quantifier{\exists}{#1}{#2}}
\newcommand{\ExistsExOne}[2]{\Quantifier{\exists!}{#1}{#2}}

\newcommand{\NaturalNumbers}{\mathbb{N}}


%% file: paper_overview_figure.tex
\begin{tikzpicture}
  \draw[very thick] (7,0) -- (-11,0);
  %
  %
  \begin{scope}[shift={(-4,0)}]    
    \draw[thick] (-3,0) parabola bend (0,1.5) (3,0);
    \node[align=center] at (-.1,1.1) {NCC};
  \end{scope}
  %
  %
  \begin{scope}[shift={(-4,0)}]
    \draw[thick] (-3,0) parabola bend (0,3.2) (3,0);
  \end{scope}
  %
  %
  \begin{scope}[shift={(-0.5,0)}]
    \draw[thick] (-7.5,0) parabola bend (0,11.6) (7.5,0);
    \node at (0,10.8) {\(\IND{\exists_{1}}\)};
  \end{scope}
  %
  %
  \begin{scope}[shift={(2,0)}]
    \draw[thick] (-5,0) parabola bend (0,6) (5,0);
    \node at (1.5,4.3) {\(\IND{\Open}\)};
  \end{scope}
  %
  %
  \begin{scope}[shift={(-4,0)}]
    \node[align=center] at (-.1,2.3) {
      \(\mathrm{CSC}\) \\
      \(=_{\text{Thm.~\ref{thm:ref_csc_logical_characterization_clause_sets}}}\) \\
      \([\varnothing,\RINDParameterFreeEta{\exists_{1}}]\)};
  \end{scope}
  %
  %
  \begin{scope}[shift={(0,0)}]
    \draw[thick] (-3,0) parabola bend (0,3.2) (3,0);
    \node at (0,2) {\(\INDParameterFree{\Open}\)};
  \end{scope}
  %
  %
  \begin{scope}[shift={(-6,0)}]
    \draw[thick] (-5,0) parabola bend (0,6) (5,0);
    \node at (-.8, 5) {\(\mathrm{IND}_{\eta}\)};
  \end{scope}
  %
  %
  \begin{scope}[shift={(-2,0)}]
    \draw[thick] (-5,0) parabola bend (0,6) (5,0);
    \node at (0, 5.3) {\([\varnothing, \RINDParameterFree{\exists_{1}}]_{1}\)};
  \end{scope}
  %
  %
  \begin{scope}[shift={(-2,0)}]
    \draw[thick] (-5,0) parabola bend (0,8.5) (5,0);
    \node[align=center] at (0, 7.5) {\(\varnothing + \RINDParameterFree{\exists_{1}}\) \\ \vdots};
  \end{scope}
  \begin{scope}[shift={(-2,0)}]
    \node at (0,6.3) {\([\varnothing,\RINDParameterFree{\exists_{1}}]_{2}\)};
    \draw[thick] (-5,0) parabola bend (0,6.9) (5,0);
  \end{scope}
  %
  %
  \node[] at (-5, .5) {\(\bullet_{1}\)};
  \node[] at (0,1) {\(\bullet_{2}\)};
  \node[] at (4,3.5) {\(\bullet_{3}\)};      
  %
  %
  %
  \begin{scope}[shift={(-2,0)}]
    \draw[dashed, thick] (-6,0) parabola bend (0,9.5) (6,0);
    \node at (0,9.0) {\(\INDParameterFree{\exists_{1}}\)};
  \end{scope}
  %
  %
  \node at (1.7, 4.9) {\(\ast_{1}\)};
  \node at (1, 4.9) {\(\bullet_{5}\)};
  \node at (0.4,4.9) {\(\bullet_{4}\)};
  \node at (-5.6,1.5) {\(\ast_{2}\)};
  %
  %
\end{tikzpicture}
